\documentclass[11pt]{article}
\usepackage{fullpage}
\usepackage{parskip}

\usepackage{amsthm,amsmath,amssymb,booktabs,xcolor,graphicx}
\usepackage{thm-restate}
\usepackage{bbm}
\usepackage{listings}
\usepackage{xcolor}
\usepackage{appendix}
\usepackage{enumerate}
\usepackage{mathbbol}
\usepackage{comment}
\usepackage[shortlabels]{enumitem}

\usepackage[margin=1in]{geometry}
\usepackage[utf8]{inputenc}
\usepackage[T1]{fontenc}

\usepackage[textsize=small,backgroundcolor=orange!20,disable]{todonotes}
\usepackage[hidelinks]{hyperref}
\usepackage{url}
\usepackage{etoolbox}
\usepackage{appendix}
\usepackage{xspace}
 \usepackage{algorithm}
\usepackage[noend]{algpseudocode}
\usepackage[labelfont=bf]{caption}
\usepackage[noabbrev,capitalize]{cleveref}
\crefname{equation}{}{} 
\AtBeginEnvironment{appendices}{\crefalias{section}{appendix}} 

\usepackage[color,final]{showkeys} 

\colorlet{refkey}{orange!20}
\colorlet{labelkey}{blue!30}

\numberwithin{equation}{section}
\newtheorem{theorem}{Theorem}[section]
\newtheorem{proposition}[theorem]{Proposition}
\newtheorem{lemma}[theorem]{Lemma}
\newtheorem{claim}[theorem]{Claim}

\crefname{claim}{Claim}{Claims}

\newtheorem{corollary}[theorem]{Corollary}

\newtheorem*{question*}{Question}
\newtheorem{fact}[theorem]{Fact}

\theoremstyle{definition}
\newtheorem{definition}[theorem]{Definition}

\newtheorem*{definition*}{Definition}

\theoremstyle{remark}

\newcommand{\poly}{\mathrm{poly}}

\newcommand{\argmin}{\mathrm{argmin }}
\newcommand{\argmax}{\mathrm{argmax }}
\newcommand{\eps}{\varepsilon}

\newcommand{\R}{\mathbb{R}}

\newcommand{\Z}{\mathbb{Z}}

\renewcommand{\d}{d}

\newcommand{\fO}{\mathcal{O}}
\newcommand{\tO}{\tilde{O}}

\newcommand{\td}{\widetilde{d}}

\newcommand{\hd}{\widehat{d}}

\newcommand{\tD}{\tilde{D}}

\newcommand{\avg}{\mathrm{avg}}

\newcommand{\congest}{$\mathsf{CONGEST}\,$}

\newcommand{\1}[1]{\mathbb{1}[#1]}

\renewcommand{\tilde}{\widetilde}
\newcommand{\lf}{\lfloor}
\newcommand{\rf}{\rfloor}

\newcommand{\Nin}{N^{\text{in}}}
\newcommand{\Nout}{N^{\text{out}}}

\title{Parallel Breadth-First Search and Exact Shortest Paths\\ and Stronger Notions for Approximate Distances}

\author{
Václav Rozhoň\thanks{Supported by the European Research Council (ERC) under the European Unions Horizon 2020 research and innovation programme (grant agreement No.~853109).}\\
\small ETH Zurich \\
\small rozhonv@inf.ethz.ch\\
\and \textcircled{r}\footnote{The author ordering was randomized using \url{https://www.aeaweb.org/journals/policies/random-author-order/generator}. 
 It is requested that citations of this work list the authors separated by \texttt{\textbackslash textcircled\{r\}} instead of commas} \and
Bernhard Haeupler \thanks{Supported in part by NSF grants CCF-1814603, CCF-1910588, NSF CAREER award CCF-1750808, a Sloan Research Fellowship, funding from the European Research Council (ERC) under the European Union's Horizon 2020 research and innovation program (ERC grant agreement 949272), and the Swiss National Foundation (project grant 200021-184735).} \\
\small Carnegie Mellon University \& ETH Zurich\\
\small bernhard.haeupler@inf.ethz.ch
\and 
Anders Martinsson \\
\small ETH Zurich \\
\small anders.martinsson@inf.ethz.ch\\
\and \textcircled{r} \and
Christoph Grunau  \footnotemark[1]\\
\small ETH Zurich \\
\small cgrunau@inf.ethz.ch\\
\and \textcircled{r} \and
Goran Zuzic\\
\small ETH Zurich \\
\small goran.zuzic@inf.ethz.ch \\
}

\date{}

\begin{document}

\maketitle

\begin{abstract}
We introduce stronger notions for approximate single-source shortest-path distances, show how to efficiently compute them from weaker standard notions, and demonstrate the algorithmic power of these new notions and transformations. One application is the first work-efficient parallel algorithm for computing exact single-source shortest paths graphs --- resolving a major open problem in parallel computing. 

\medskip

Given a source vertex in a \emph{directed graph} with polynomially-bounded nonnegative integer \emph{lengths} the algorithm computes an \emph{exact shortest path} tree in $m \log^{O(1)} n$ work and $n^{1/2+o(1)}$ depth. %
Previously, no parallel algorithm improving the trivial linear depths of Dijkstra's algorithm without significantly increasing the work was known, even for the case of undirected and unweighted graphs (i.e., for computing a BFS-tree). 

\medskip

Our main result is a black-box transformation that uses $\log^{O(1)} n$ standard approximate distance computations to produce approximate distances which also satisfy the subtractive triangle inequality (up to a $(1+\eps)$ factor) and even induce an exact shortest path tree in a graph with only slightly perturbed edge lengths. These strengthened approximations are algorithmically significantly more powerful and overcome well-known and often encountered barriers for using approximate distances. In directed graphs they can even be boosted to exact distances. This results in a black-box transformation of any (parallel or distributed) algorithm for approximate shortest paths in directed graphs into an algorithm computing exact distances at essentially no cost. Applying this to the recent breakthroughs of Fineman et al. for compute approximate SSSP-distances via approximate hopsets gives new parallel and distributed algorithm for exact shortest paths.

\end{abstract}

\thispagestyle{empty}
\setcounter{page}{0}

\newpage

\tableofcontents



\newpage

\setcounter{page}{1}

\section{Introduction}

This paper gives the first \emph{work-efficient} parallel algorithms with \emph{sublinear depth} for computing the \emph{exact} distances and shortest paths in weighted graphs --- resolving major open problems in parallel computing and breaking a long-standing barrier where no such algorithm was known even for breadth-first search (BFS). The simplest version of our main results are easy to state formally:

\begin{theorem}\label{thm:main-exact-intro}
There exists a parallel algorithm that takes as input any directed or undirected graph $G=(V,E)$ with $m=|E|$ edges, $n=|V|$ nodes, a length function $\ell$ that assigns to every edge $e$ a non-negative polynomially-bounded integer lengths $\ell(e)$, and any source vertex $s \in V$. The algorithm computes for every vertex the exact distance to $s$ and an exact shortest path tree rooted at $s$. The algorithm is randomized, requires $m \log^{O(1)} n$ work, $n^{1/2+o(1)}$ depth, and works with high probability.
\end{theorem}

Note that computing a breadth-first-search (BFS) tree corresponds to the much simpler special case of a shortest-path tree in an undirected graph where all edge lengths are one. 

We develop our results in a surprisingly indirect way: we develop \emph{stronger notions of approximate distances} that are algorithmically more powerful than the standard notions used throughout the literature.
We give efficient black-box reductions computing these stronger approximations using only polylogarithmic black-box calls to algorithms computing the (weaker) \emph{standard} distance approximations. 

For directed graphs, our stronger approximations can be efficiently boosted to exact distances \cite{klein1997randomized}. Using the parallel shortest-path approximation algorithm of \cite{cao2020paralleldirected} in our reductions yields \Cref{thm:main-exact-intro}. 

For undirected graphs, much faster work-efficient parallel approximation algorithms with polylogarithmic depth are known~\cite{Li20shortest_paths,andoni_stein_zhong2020shortest_paths,rozhon_grunau_haeupler_zuzic_li2022deterministic_sssp}. While we do not know how to use these results to develop fast and work-efficient exact distance algorithms on undirected graphs, we show how to compute the next best thing: an exact shortest path tree in a slightly perturbed graph.

\begin{theorem}\label{thm:main-undirected-full-intro}
There exists a parallel algorithm that takes as input any undirected graph $G=(V,E)$ with $m=|E|$ edges, $n=|V|$ nodes, a length function $\ell$ that assigns to every edge $e$ a non-negative polynomially-bounded integer lengths $\ell(e)$, any source vertex $s \in V$, and any $\eps \in (0,1]$. The algorithm computes a $(1+\eps)$-perturbed length function $\ell'$, which satisfies $\ell(e) \leq \ell'(e) \leq (1+\eps) \ell(e)$ for every edge $e \in E$, and outputs $\ell'$ together with exact distances to $s$ and an shortest path tree rooted at $s$ under this new length function $\ell'$. The algorithm is deterministic, requires near-linear work of $m \cdot \left(\frac{\log n}{\eps}\right)^{O(1)}$, and polylogarithmic depth of $\left(\frac{\log n}{\eps}\right)^{O(1)}$.
\end{theorem}

While our headline result is for directed graphs, we expect the newly-developed techniques --- stronger notions of approximate distances together with the near-optimal algorithms to compute them --- to be conceptually and algorithmically useful even in undirected graphs. Specifically, many (undirected) graph-theoretic structures and algorithms are crucially built on top of exact distances and our concepts often allow approximate distances (satisfying our additional guarantees) to be used. This enables leveraging state-of-the-art parallel algorithms for approximate distances which are work-efficient and have near-optimal polylogarithmic depth, a quality that seems out of reach if one requires exact distances.

\subsection*{Organization}

We give some motivation and background on parallel algorithms for the breadth-first-search and shortest path problems in \Cref{sec:history}, and give a brief summary of related work in \Cref{sec:prior_work}. We explain key insights in the subtle problems making standard distance approximations hard to use algorithmically in \Cref{sec:issues}. These issues directly motivate the definitions of our stronger notions of distance approximation given in \Cref{sec:stronger_notions}. \Cref{sec:our_results} gives a more detailed account of our results. This is followed by the technical part of the paper starting with preliminaries in \Cref{sec:preliminaries}. 

\subsection{The Breadth-First Search and Shortest-Path Problems}
\label{sec:history}

Big-data analytics on petabyte-sized network representations and data-intensive computations on graphs have become pervasive and their scale keeps exploding exponentially. Indeed, many applications involve graph abstractions with billions of edges and  massively-parallel systems with millions of cores working on these graphs. 

Unfortunately even very basic graph algorithmic sub-routines can be notoriously hard to parallelize. The most poignant and common example is simple breadth-first search (BFS). While a linear-time BFS implementation for a single core taking $O(m)$ time to process any $m$-edge graph is trivial, no parallel algorithm reducing the time using more cores without drastically increasing the total work is known --- even for undirected graphs. 

The complete lack of algorithms with provable guarantees for this basic problem is particularly shocking given the decades-long continued theoretical and practical interest in this problem~\cite{hart1968formal,ullman1991high,spencer1991more,thorup1992shortcutting,subramanian1995efficient,cohen1996efficient,cohen2000polylog,thorup2005approximate,fu2006heuristic,shun2013ligra,madduri2007experimental,dhulipala2017julienne,abraham2010highway,ueno2012highly,satish2014navigating,fineman2018nearly,cao2020improved,cao2020paralleldirected,Cao2021ImprovedCONGEST,kainer2019more,elkin2019hopsets,Li20shortest_paths,andoni_stein_zhong2020shortest_paths,huang2021lower,rozhon_grunau_haeupler_zuzic_li2022deterministic_sssp,kogan2022new}. For example, the well-known Graph500 benchmark~\cite{graph500benchmark}, which is used to rank supercomputers, consists of merely two problems, firstly computing a BFS tree on graphs with up to $2^{45}$ nodes, and secondly computing single-source shortest paths (SSSP), i.e., essentially a weighted BFS.

The single-source shortest-path problem (SSSP) asks to compute the minimum-length path from a source vertex to all other vertices in a graph with positive edge lengths. It is one of the most important optimization problems in computer science and parallel SSSP algorithms have been intensely studied for decades but also particularly recently, owing to the increased practical need for graph algorithms suitable for modern massively-parallel systems. 

While Dijkstra's classic 1956 algorithm~\cite{dijkstra1959dijkstra} computes \emph{exact} distances and SSSP trees in directed graphs with lengths in essentially linear work, it does so in an inherently sequential way. Again, the only known algorithms improving over the trivial linear depth of Dijkstra's algorithm do so at the cost of significantly increasing the total work and the best known time-work trade-offs remain the algorithms of Ullman \& Yannakakis~\cite{ullman1991high} and Spencer~\cite{spencer1991more} from the 90's with $\tilde{O}(m\rho + \frac{\rho^4}{n})$ and $\tilde{O}(m + n\rho^2)$ work for a depth of $O(n/\rho)$ for any $\rho \in [1,n]$.

\subsection{Prior Work and Approximate SSSP Algorithms}
\label{sec:prior_work}

Given the difficulty of designing parallel BFS and SSSP algorithms, research and algorithm design efforts have turned to heuristics~\cite{hart1968formal,fu2006heuristic} and algorithms with provable guarantees for restricted graph classes or other special cases, including planar graphs~\cite{henzinger1997faster,fakcharoenphol2006planar}, random graphs~\cite{frieze1984parallel,meyer1998delta}, or graphs arising in road networks~\cite{abraham2010highway}. Another intensely-pursued research direction has been the design of parallel and distributed algorithms that compute approximate instead of exact distances~\cite{cohen2000polylog,rozhon_grunau_haeupler_zuzic_li2022deterministic_sssp,goranci2022universally,andoni_stein_zhong2020shortest_paths,Li20shortest_paths,haeupler2018faster,becker2021near}. Approximation algorithms for SSSP are most closely related and relevant to this paper and we focus the remaining discussion of related work on this direction. 

The intuition why approximate distances are much easier to compute is that they break up the rigid long-range interactions that seem to make BFS or exact SSSP computations so inherently sequential and hard to parallelize. In particular, even the most minor length or distance change on a single node or edge potentially causes not just all vertices its subtree to require distance updates but can completely changing the structure of the BFS or SSSP tree. Even for (unit length) undirected graphs, this makes it very hard to do any meaningful BFS computations in parts of a graph far away from the source without being sure which node will be reached first as even the most minor change can completely change the outcome. In the approximate setting one can often avoid such rigid long-range interactions by computing approximate solutions that have build-in slacks that can absorb minor updates or changes. 

Many ideas have been brought forward to utilize this flexibility. Overall, the quest for parallel and distributed SSSP approximation algorithms has been a tremendous inspiration for innovation resulting in a wide variety of fundamental and powerful paradigms and structures that have become a crucial part of the modern algorithm design toolbox. Among many others these include (approximate) hopsets~\cite{cohen2000polylog,elkin2019hopsets}, spanners and emulators~\cite{thorup2006spanners,huang2019thorup}, low-diameter decompositions~\cite{miller2013parallel,blelloch14,elkin_haeupler_rozhon_grunau2022Clusterings_LSST}, diameter-reducing shortcuts~\cite{thorup1992shortcutting,huang2021lower,kogan2022new}, and distance oracles~\cite{thorup2005approximate,chechik2014approximate}. The search for better parallel approximation algorithms for the shortest-path problem has also lead to a mutually beneficial exchange of ideas and cross-fertilization between continuous and discrete optimization. With research borrowing and further developing ideas from continuous optimization, such as the multiplicative weight methods~\cite{Li20shortest_paths,andoni_stein_zhong2020shortest_paths,becker2021near,rozhon_grunau_haeupler_zuzic_li2022deterministic_sssp,goranci2022universally}. Many of the above tools have been specifically developed for parallelizing shortest-path computations and all state-of-the-art work-efficient parallel algorithms for computing approximate distances crucially utilize many of these tools.


For undirected graphs the state of the art is given by the algorithms of  \cite{rozhon_grunau_haeupler_zuzic_li2022deterministic_sssp,Li20shortest_paths,andoni_stein_zhong2020shortest_paths} and for directed graphs the algorithm of Cao, Fineman, and Russel~\cite{cao2020paralleldirected}. These algorithms all compute $(1+\eps)$-approximations using near-linear work of $\tilde{O}(\eps^{-O(1)} m)$. The algorithm for undirected graphs achieve a near-optimal polylogarithmic depth of $(\frac{\log n}{\eps})^{O(1)}$ and the algorithm of \cite{cao2020paralleldirected} for directed graphs has depth $n^{1/2+o(1)}$. 

These algorithms of \cite{rozhon_grunau_haeupler_zuzic_li2022deterministic_sssp} and \cite{cao2020paralleldirected} also have distributed message-passing implementations (i.e., in the standard CONGEST model of distributed computing) --- another intensely studied related area of work with tremendous recent activity and progress.

We conclude our brief summary by noting that basically all existing tools and algorithms seem inherently limited to approximation algorithms (with at best a polynomial dependency on $\eps$ for $(1+\eps)$-approximations). Moreover, the vast majority of structures, tools, and algorithms seem inherently restricted to undirected graphs. These limitations explains the complete lack of work-efficient parallel or distributed algorithms for exact shortest paths as well as the large gap between what is known for approximation algorithms in undirected graphs versus directed graph.


\subsection{Issues with Standard Notions of Distance Approximations: Monotonicity and Subtractive Triangle Inequality}
\label{sec:issues}

While approximate distances provide flexiblities that makes them easier to compute, these flexibilities cause several subtle issues if we try to  use approximate distances as algorithmic tools. This holds even for quite precise $(1 + \eps)$-approximations. We discuss these issues next.

The most immediate notion of approximation when discussing the SSSP problem on a graph $G$ is to ask for a distance estimate $( \td(v) )_v$ that $\alpha$-approximates the exact distances from the source to each node $v \in V(G)$. We formalize the notion below.

\begin{definition}[Shortest-Path Distances]
Let $G=(V,E,\ell)$ be a graph with edge lengths $\ell: E \rightarrow \R_{\ge 0}$. We denote with $d_G(u,v)$ the shortest-path distance between $u$ and $v$ in the graph $G$.
\end{definition}


\begin{definition}[$\alpha$-Approximate Distance Estimate]\label{def:approx-dist}
  An arbitrary function $\td : V(G) \to \R_{\ge 0}$ is called a distance estimate. Given a graph $G = (V, E, \ell)$ and a node $s \in V$, a distance estimate $\td: V \rightarrow \R_{\ge 0}$ is $\alpha$-approximate (with respect to the source $s$) if the following two conditions are satisfied:
  \begin{align}
  \label{eq:noncontractivity}
      \forall u \in V: \ \ d_G(s, u) \le \td(u) 
  \end{align}
  \begin{align}
  \label{eq:approximation}
      \forall u \in V: \ \ \td(u)  \le \alpha \cdot d_G(s, u). 
  \end{align}
  We sometimes refer to the first condition as noncontractivity. 
  When $\alpha = 1$, the inequalities become equalities and we then call $\td$ \emph{exact} distances (from $s$).
\end{definition}

Unfortunately, while $(1+\eps)$-approximations seem might seem ``almost as good as exact'' on first glance, replacing exact distance computations by approximate ones often completely breaks algorithms and correctness proofs in subtle and unexpected ways. This is already true when restricted to unweighted and undirected graphs. We identify two main culprits that break many algorithms which build upon (exact) shortest path computations if one were to replace exact computations with $(1+\eps)$-approximate ones: subtractive triangle inequality and non-monotonicity of distances along shortest paths.

\textbf{Issue: subtractive triangle inequality.} In particular, if we label each node $u$ with its exact distance $d(u)$ to $s$, then for any edge $e = \{u,v\}$ we have that $|d(u)-d(v)| \leq 1$ in unweighted graphs and at most $\ell(e)$ in weighted graphs. For $(1+\eps)$-approximate distances $\td$, the value $|\td(u) - \td(v)|$ can essentially be unboundedly large (e.g., as large as $\eps n$ in unweighted graphs). The subtractive triangle inequality is crucially used in many algorithms: examples include the MPX algorithm~\cite{miller2013parallel} for low-diameter decompositions (see \Cref{subsec:mpx} for a discussion), boosting approximate to exact distances in directed graphs~\cite{klein1997randomized}, or Bourgain's algorithm for embedding a graph into the $\ell_1$ metric with low distortion~\cite{linial1995geometry}.


\textbf{Issue: non-monotonicity.} The other problem is lack of monotonicity along shortest paths. In particular, approximate distances can have local minima and can be non-monotone along every single path from $s$ to $u$. We illustrate these issues on the following simple example.

\textbf{Example: ball growing.} Let $G$ be an undirected (and even unweighted) graph, $s \in V(G)$ be a node, and $R > 0$ be some radius. Let $r$ be uniformly randomly chosen in $[R/2,R]$ and let $B(s, r) = \{u \in V \mid d_G(s, u) \leq r\}$ be the ball around $s$ with radius $r$. It is easy to see that $B$ is indeed a ``nice ball'': e.g., the induced subgraph $G[B]$ is connected and has diameter $r$. Furthermore any (unit) edge $\{u, v\}$ in the graph is cut by this ball with probability at most $\frac{|d_G(s, u) - d_G(s, v)|}{R/2} = O(1/R)$, giving us a convenient way to bound the number of cut edges.

If we replace $d$ by a $(1+\eps)$-approximate distance $\td$ then $B$ can be disconnected and connected components in $G[B]$ can have arbitrarily large diameter due to the non-monotonicity. There is also no guarantee that an edge is cut with a small probability: an edge $\{u, v\}$ is cut with probability at most $\frac{|\td(u) - \td(v)|}{R/2}$ which can be as large as $\Theta(\frac{\eps n}{R})$ due to the absence of triangle inequality. Hence, the expected number of cut edges might be much larger compared to the exact case. 

These above issues with approximate distances in general and the ball growing example in particular have been noted and pointed out in many prior works on approximate SSSP including among others in \cite{becker_emek_lenzen2020blurry_ball_growing,forster2018faster,Li20shortest_paths}. The most explicit treatment is the paper by Becker, Emek, and Lenzen~\cite{becker_emek_lenzen2020blurry_ball_growing} which solely and  explicitly concerned with the problems the lack of an (approximate) subtractive triangle inequality causes. They show, with lot of effort, that weaker versions of the above ``ball cutting'' can be achieved algorithmically by designing a randomized ``blurry ball growing'' technique which can tolerate approximate distances. A slightly simpler and deterministic equivalent of this was presented in \cite{elkin_haeupler_rozhon_grunau2022Clusterings_LSST}. The notions and algorithms of this paper can be seen as a more  general and powerful way to solve problems like the above (see \Cref{subsec:level_cuts} for details).

\subsection{Stronger Notions of Distance Approximations}
\label{sec:stronger_notions}

An important conceptual contribution of this paper is the analysis of the following strengthened notions of approximate distance aimed squarely at avoiding the just discussed issues. We first introduce two properties, smoothness and tree-likeness, that strengthen the conditions \Cref{eq:approximation} and \Cref{eq:noncontractivity} in the definition of  $\alpha$-approximate distance estimates in \Cref{def:approx-dist}. 

\begin{definition}[Smoothness and Tree-Likenes] \label{def:smooth+treelike}
  Given a graph $G = (V, E, \ell)$ and (so-called source) $s \in V$, we say a distance estimate $\td : V \to \R_{\ge 0}$ is:
  \begin{itemize}
  \item \textbf{$\alpha$-smooth} (w.r.t. $s$) iff both $\td(s) = 0$ and $\forall u,v \in V: \ \ \td(v) - \td(u) \leq \alpha \cdot d_G(u,v)$,
  \item \textbf{tree-like} (w.r.t. $s$) iff both $\td(s) = 0$ and $\forall v \neq s: \ \exists e=(u,v) \in E: \ \ \td(u) \leq \td(v) - \ell(e)$.
  \end{itemize}
\end{definition}

Next, we define the following strengthened notions of $\alpha$-approximate distances. 
\begin{definition}[$\alpha$-approximate distances] \label{def:strongly}
  Given a graph $G = (V, E, \ell)$ and (so-called source) $s \in V$, we say a distance estimate $\td : V \to \R_{\ge 0}$ is:
  \begin{itemize}
  \item \textbf{smoothly $\alpha$-approximate} (w.r.t. $s$) if we replace \Cref{eq:approximation} in \Cref{def:approx-dist} by the $\alpha$-smooth condition. 
  \item \textbf{strongly $\alpha$-approximate} (w.r.t. $s$) if we replace both \Cref{eq:noncontractivity} by the tree-like condition and \Cref{eq:approximation} by the $\alpha$-smooth condition in \Cref{def:approx-dist}. 
  \end{itemize}
\end{definition}

An example showing how smoothly $\alpha$-approximate distances strengthen $\alpha$-approximate distances is given in \Cref{fig:example}. \Cref{fig:smooth_path} shows another diagram explaining the definitions from \Cref{def:strongly}.

\begin{figure}
    \centering
    \includegraphics[width = \textwidth]{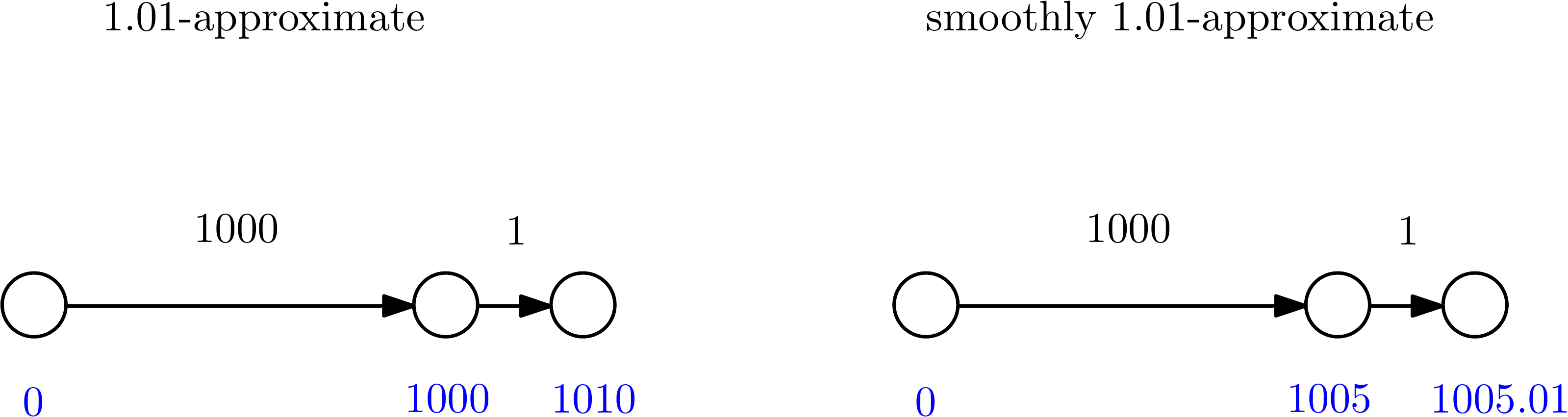}
    \caption{Example showing possible $\alpha$-approximate and smoothly $\alpha$-approximate distance estimates (estimates are blue). Note that in the $\alpha$-approximate case (left picture), the distance estimate can increase abruptly over a short edge. This is forbidden if the estimate is smoothly $\alpha$-approximate (right picutre). }
    \label{fig:example}
\end{figure}

\begin{figure}
    \centering
    \includegraphics[width = \textwidth]{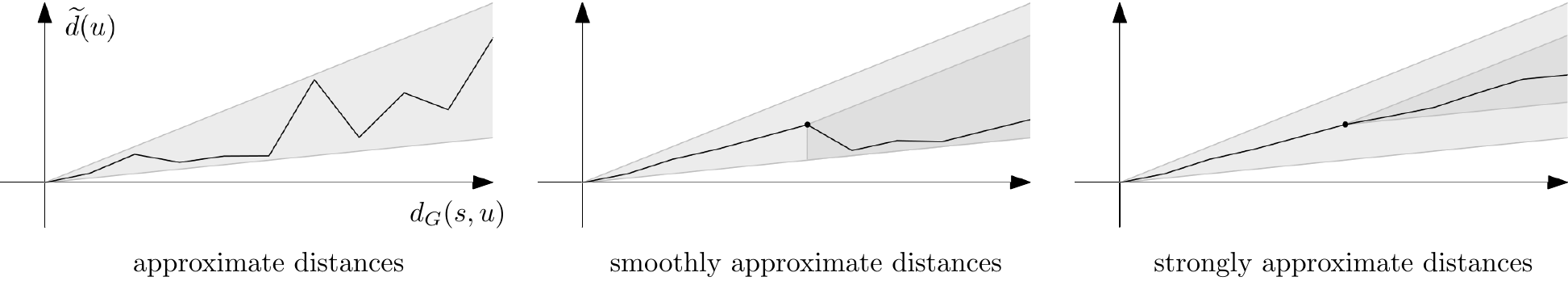}
    \caption{The diagram shows an example of approximate distance estimates in case the input graph is an unweighted oriented path and the source is its first node. 
    In all three cases we highlight two rays of slopes $1$ and $\alpha$, distance estimates need to lie within the gray cone that the rays define. \\
    1) In the first picture ($\alpha$-approximate estimate), the values are arbitrary within the grey cone. \\
    2) In the second picture (smoothly $\alpha$-approximate estimate), the distances cannot increase abruptly. Namely, the values of distances after the higlighted node can lie only in the darker area. Distance estimate can, however, drop abruptly. \\
    3) In the third picture (strongly $\alpha$-approximate estimate), the distances cannot decrease abruptly and they also increase at least at the rate of exact distances (see the cone after a highlighted node). Observe that if we perturb the edge lengths of the input graph multiplicatively by at most $\alpha$, the distance estimate becomes exact (see \Cref{lem:smooth+treelike=cool}).  }
    \label{fig:smooth_path}
\end{figure}

We remark that for our results for directed graphs (e.g., \Cref{thm:main-exact-intro}), only the $\alpha$-smoothness property is required. On the other hand, for our results in undirected graphs it is crucial that a single distance function simultaneously satisfies both properties of a strongly $\alpha$-approximate distance simultaneously. Only then does one obtain the powerful characterization of strongly $\alpha$-approximate distances used in \Cref{thm:main-undirected-full-intro}, i.e., strongly $\alpha$-approximate approximate distances are exactly equivalent to exact distances induced by a slightly perturbed edge-length function. A detailed discussion and formal treatment of these notions and how they compare, e.g., to standard notions of approximation, is provided in \Cref{sec:distances}.



\subsection{Our Results}
\label{sec:our_results}

Our main technical contribution is a procedure for computing smoothly $(1+\eps)$-approximate distance estimates using $O(\log n)$ calls to an approximate distance oracle.
The procedure works both for directed and undirected graphs. However, the implications of this transformation are quite different in the directed and undirected setting. We therefore split the result section into a directed and undirected part. Our reductions from this section are also explained in \Cref{fig:definitions}. 

\begin{figure}
    \centering
    \includegraphics[width = \textwidth]{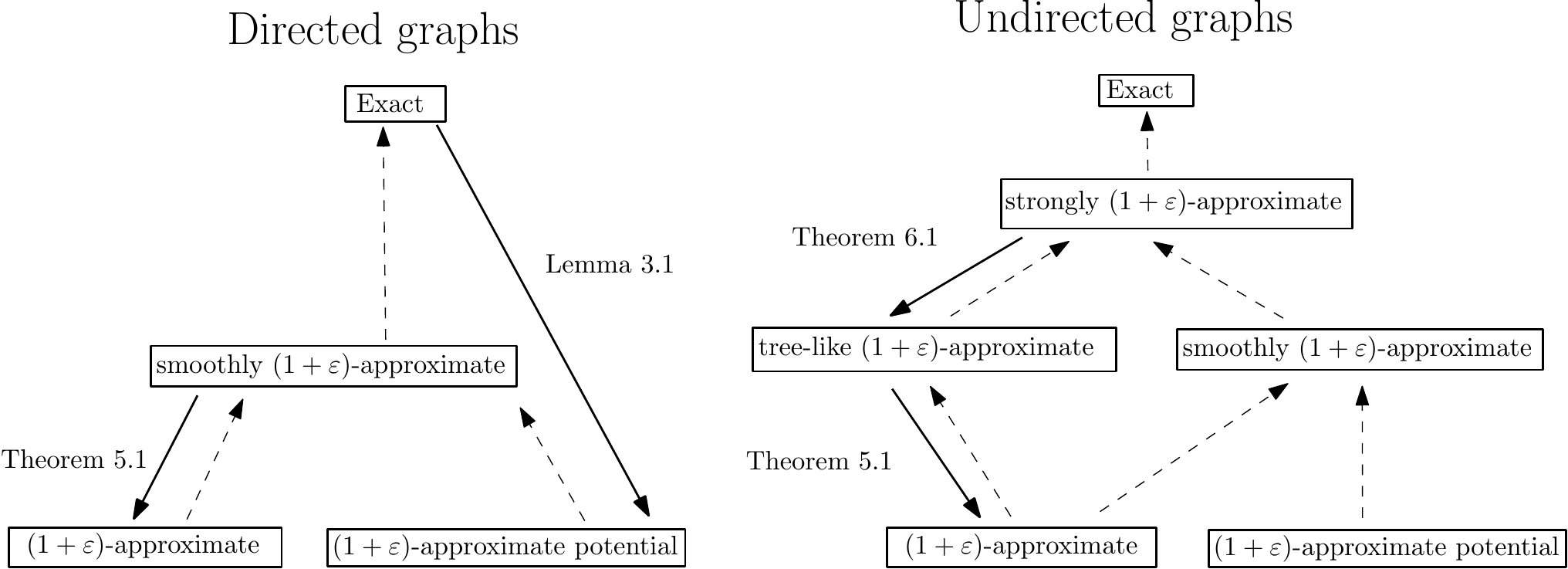}
    \caption{The figure shows the relations between the various stronger approximate distance notions that we consider in this paper. The dashed arrows show the trivial inclusions while the solid arrows show reductions of stronger notions to simpler ones. \\
    \textbf{Directed graphs:} \cite{klein1993linear} (\Cref{lemma:folklore_directed_smoothing}) shows how to reduce the exact distance computation to the computation of approximate potential. 
    Since smoothly approximate distances strengthen approximate potentials, our reduction from smoothly approximate distances to approximate distances in \Cref{thm:smooth} then reduces the exact distance computation to approximate one. \\
    \textbf{Undirected graphs:} 
    We define the notion of strongly approximate distances that combine smoothness with tree-likeness. \Cref{thm:smooth} reduces the computation of strongly approximate distances to the computation of tree-like approximate distances. Next, \Cref{thm:tree} reduces the latter problem further down to the problem of computing approximate distances. 
    }
    \label{fig:definitions}
\end{figure}

\subsubsection{Directed Graphs}

Our main result is an algorithm that turns approximate distances into smoothly approximate ones. 

\begin{theorem}[Computing smoothly $(1+\eps)$-approximate distance estimates in directed graphs]
 \label{theorem:results_directed_smooth}
 Given a directed graph $G$ with (real) weights in $[1,\text{poly}(n)]$, a source $s \in V(G)$, and accuracy $\eps \in (0,1]$, \Cref{alg:smoothing} computes smoothly $(1+\eps)$-approximate distance estimates from $s$ in $G$ using $O(\log n)$ calls to a $(1 + O(\eps/\log n))$-approximate distance oracle $\fO$ on directed graphs. 
\end{theorem}

By using a reduction of Klein and Subramanian \cite{klein1997randomized}, one can compute exact distances with $O(\log n)$ calls to a smoothly $2$-approximate distance estimate oracle.
To make this paper self-contained and to better match the language we use in this paper, we give a self-contained treatment of this reduction in \Cref{sec:boosting-sssp}.
By combining this reduction from exact distances to smoothly approximate distance estimates with our reduction from smoothly approximate distance estimates to approximate distance estimates, we obtain an efficient reduction reducing exact distances to the computation of $(1 + O(1/\log n))$-approximate distance estimates, modulo some low-level technical details that we handle and discuss in more detail in \Cref{sec:applications_computational_model}. This reduction allows us to turn any parallel algorithm for computing $(1 + O(1/\log n))$-approximate distance estimates in directed graphs into a parallel algorithm for computing exact distances with almost no overhead. 
In a recent breakthrough, Cao, Fineman and Russell gave a parallel algorithm for computing $(1 + O(1/\log n))$-approximate distances with $\tO(m \rho^2 + n \rho^4)$ work and $n^{1/2 + o(1)}/\rho$ depth, for a given parameter $\rho \in [1,n^{1/2}]$ \cite{cao2020paralleldirected}. We therefore obtain as a corollary the following result.

\begin{restatable}{corollary}{directedpram}[Directed SSSP, PRAM]
    \label{cor:results_directed_pram}
    There exists a parallel algorithm that takes as input any directed graph $G$ with nonnegative polynomially bounded integer edge lengths, any source vertex $s \in V(G)$ and a parameter $\rho \in [1,n^{1/2}]$. The algorithm computes for every vertex the exact distance to $s$ and an exact shortest path tree rooted at $s$. The algorithm is randomized, requires $\tilde{O}(m \rho^2 + n \rho^4)$ work and $n^{1/2+o(1)}/\rho$ depth, and works with high probability.
\end{restatable}

In a similar vein, using the state-of-the-art \congest algorithm for computing $(1+O(1/\log n))$-approximate distance estimates by Cao, Fineman and Russell \cite{Cao2021ImprovedCONGEST} and observing that their algorithm works if the source node is a virtual node connected to all other nodes in the graph, we obtain the following result.

\begin{restatable}{corollary}{directedcongest}[Directed SSSP, CONGEST]
    \label{cor:results_directed_congest}
    There exists a distributed algorithm that takes as input any directed graph $G$ with nonnegative polynomially bounded integer edge lengths and any source vertex $s \in V(G)$. The algorithm computes for every vertex the exact distance to $s$ and an exact shortest path tree rooted at $s$. The algorithm is randomized, runs in $\tO(n^{1/2} + D + n^{2/5 + o(1)}D^{2/5})$ rounds where $D$ denotes the undirected hop-diameter of $G$, and works with high probability.
\end{restatable}

For some range of parameters, this result improves upon the previous fastest exact shortest path \congest algorithm with a round complexity of $\tilde{O}(\sqrt{n}D^{1/4} + D)$ \cite{chechik2020single} by a polynomial factor. 

\subsubsection{Undirected Graphs}

We next present our results on undirected graphs. As mentioned above, our procedure for computing smoothly $(1+\eps)$-approximate distance estimates also works for undirected graphs. Moreover, the reduction preserves tree-likeness. That is, if the oracle returns tree-like approximate distances, then the final distance estimate will also be tree-like, on top of being smoothly $(1+\eps)$-approximate, and therefore by definition strongly $(1+\eps)$-approximate.

\begin{theorem}[Computing smoothly $(1+\eps)$-approximate distance estimates in undirected graphs while preserving tree-likeness]
\label{thm:results_undirected_smooth}
 Given an undirected graph $G$ with (real) weights in $[1,poly(n)]$, a source $s \in V(G)$, and accuracy $\eps \in (0,1]$. 
 \Cref{alg:smoothing} computes smoothly $(1+\eps)$-approximate distance estimates from $s$ in $G$ using $O(\log n)$ calls to a $(1 + O(\eps/\log n))$-approximate distance oracle $\fO$ on undirected graphs. 
 Moreover, if $\fO$ returns tree-like $(1+\eps)$-approximate distance estimates from $s$ in $G$, then $\Cref{alg:smoothing}$ computes strongly $(1+\eps)$-approximate distance estimates from $s$ in $G$.
\end{theorem}
We note that we can get a similar result in the more complicated all pairs shortest paths setting, we leave the statement and the proof to \cref{sec:APSP}. 

We also give a procedure for computing tree-like $(1+\eps)$-approximate distance estimates which only relies on approximate distance estimates.

\begin{theorem}[Computing tree-like $(1+\eps)$-approximate distance estimates in undirected graphs]
\label{thm:results_undirected_tree_like}
Given an undirected graph $G$ with (real) weights in $[1,poly(n)]$, a source $s \in V(G)$, and accuracy $\eps \in (0,1]$. 
\Cref{alg:tree_constructing} computes tree-like $(1+\eps)$-approximate distance estimates from $s$ in $G$ using $O(\log^2 n)$ calls to a $(1 + O(\eps/\log n))$ approximate distance oracle $\fO$ on undirected graphs.
\end{theorem}

By combining the two reductions, we obtain a procedure for turning approximate distance estimates into strongly approximate distance estimates, modulo some low-level details that we discuss in \Cref{sec:applications_computational_model}.

By using the recent deterministic parallel and distributed approximate shortest path algorithms of Rozhon \textcircled{r} Grunau \textcircled{r} Haeupler \textcircled{r} Zuzic \textcircled{r} Li \cite{rozhon_grunau_haeupler_zuzic_li2022deterministic_sssp}, we then obtain the following two corollaries.

\begin{restatable}{corollary}{undirectedpram}[Undirected Strong Distance Estimates, PRAM]
\label{cor:results_undirected_pram}
  There exists a parallel algorithm that takes as input any undirected graph $G$ with nonnegative polynomially bounded integer edge lengths, any source vertex $s \in V(G)$ and a parameter $\eps \in (0,1]$. The algorithm computes a strongly $(1+\eps)$-approximate distance estimate from $s$. The algorithm is deterministic, requires $m \cdot \left(\frac{\log n}{\eps}\right)^{O(1)}$ work and depth $\left( \frac{\log n}{\eps}\right)^{O(1)}$.
\end{restatable}

\begin{restatable}{corollary}{undirectedcongest}[Directed Strong Distance Estimates, CONGEST]
\label{cor:results_undirected_congest}
  There exists a distributed algorithm that takes as input any undirected graph $G$ with nonnegative polynomially bounded integer edge lengths, any source vertex $s \in V(G)$ and a parameter $\eps \in (0,1]$. The algorithm computes a strongly $(1+\eps)$-approximate distance estimate from $s$. The algorithm is deterministic and runs in $(\sqrt{n} + D)\left(\frac{\log n}{\eps}\right)^{O(1)}$ rounds.
\end{restatable}

\section{Related Work}\label{sec:relatedwork}

The related work for the shortest path problem is vast, hence this section only tries to focus on the most relevant papers.


\textbf{Exact and Directed SSSP.} The current state of literature suggests that techniques for computing exact techniques are intricately tied to directed graphs (up to gaps that we close in this paper). For the sequential case, Dijkstra's algorithm gives an essentially-optimal solution. In the parallel setting, work-span trade-offs were given by Ullman\&Yannakakis~\cite{ullman1991high} and Spencer~\cite{spencer1991more} from the 90's with $\tilde{O}(m\rho + \frac{\rho^4}{n})$ and $\tilde{O}(m + n\rho^2)$ work for a depths of $O(n/\rho)$ for any $\rho \in [1,n]$. This was improved by Klein and Subramanian~\cite{klein1997randomized} who gave a parallel $\tilde{O}(\sqrt{n})$-depth and $\tilde{O}(m\sqrt{n})$-work exact SSSP algorithm by showing how to boost the approximate variant of the problem to the exact one. More recently, there have been improvements in construction for directed hopsets~\cite{fineman2018nearly,jambulapati2019parallel,cao2020paralleldirected}, an important combinatorial structure used for most parallel and distributed algorithms, culminating in a parallel $(1+\eps)$-approximate SSSP with work-span tradeoff of $\tilde{O}(m \rho^2 + n\rho^4)$ work and $n^{1/2+o(1)} / \rho$ depth for $\rho \in [1, \sqrt{n}]$. Exciting progress in directed hopset construction started by Kogan and Parter~\cite{kogan2022new,kogan_et_al:LIPIcs.ICALP.2022.82,bernstein2022closing} has recently broken the $\sqrt{n}$-threshold, with their construction being the bottleneck that prevents them from improving on the state-of-the-art.

In the distributed setting (specifically, for the standard message-passing model CONGEST), the first sublinear algorithm for exact directed SSSP was given by Elkin\todo{add citation} with a $\tilde{O}(n^{5/6} + n^{2/3} D^{1/3})$-round CONGEST algorithm, where $D$ is the hop-diameter of the graph. Forster and Nanongkai~\cite{forster2018faster} gave two algorithms taking $\tilde{O}(\sqrt{nD})$ and $\tilde{O}(\sqrt{n}D^{1/4} + n^{3/5} + D)$ rounds, while Ghaffari and Li~\cite{ghaffari2018improved} gave a $\tilde{O}(n^{3/4+o(1)} D^{1/6} + D)$-round and $\tilde{O}(n^{6/7} + D)$-round algorithms, both improving on Elkin's bound. Chechik and Muhtar~\cite{chechik2020single} gave a $\tilde{O}(\sqrt{n} D^{1/4} + D)$-round algorithm for exact directed SSSP. Finally, for directed $(1+\eps)$-approximate SSSP, Cao, Russell, and Fineman~\cite{Cao2021ImprovedCONGEST} gave a $\tilde{O}(\sqrt{n} + n^{2/5+o(1)} D^{2/5} + D)$-round algorithm.

\textbf{Approximate Undirected SSSP.} The progress for approximate distance computation on undirected graphs has significantly outpaced the progress on directed or exact ones. The breakthrough work of Cohen~\cite{cohen2000polylog} gave an $m^{1 + o(1)}$-work and $\poly(\log n)$-depth algorithm. This was the first truly sublinear work-efficient parallel algorithm (somewhat mirroring our result), and was achieved by introducing the (undirected) \emph{hopset}, a combinatorial structure that spurred significant research interest~\cite{elkin2019hopsets,huang2019thorup,ben2020new,elkin2020near}. The state-of-the-art for $(1+\eps)$-SSSP in PRAM is given by the randomized algorithms of Andoni etc~\cite{andoni_stein_zhong2020shortest_paths} and Li~\cite{Li20shortest_paths} for undirected graphs with $\tilde{O}(\frac{1}{\eps^2} m)$ work and $(\frac{\log n}{\eps})^{O(1)}$ depth. A deterministic algorithm with the same guarantees was given by Rozhon \textcircled{r} Grunau \textcircled{r} Haeupler \textcircled{r} Zuzic \textcircled{r} Li \cite{rozhon_grunau_haeupler_zuzic_li2022deterministic_sssp}. In the distributed setting (specifically, for the standard message-passing model CONGEST), Becker, Forster, Karrenbauer, and Lenzen~\cite{becker2021near} gave an existentially-optimal $\tilde{O}(D + \sqrt{n})$ algorithm for $(1+\eps)$-SSSP, where $D$ is the hop-diameter of the graph. This was improved to a universally-optimal algorithm (one that is as fast as possible on any network topology) by Zuzic \textcircled{r} Goranci \textcircled{r} Ye \textcircled{r} Haeupler \textcircled{r}Sun~\cite{goranci2022universally}, and made deterministic in the aforementioned \cite{rozhon_grunau_haeupler_zuzic_li2022deterministic_sssp}.

\textbf{Work-efficient parallel algorithms in special graph classes.} For Erdos-Renyi random graphs, Crauser, Mehlhorn, Meyer, and Sanders~\cite{crauser1998parallelization} gave a work-efficient $\tilde{O}(n^{1/3})$-depth exact SSSP. For planar and genus-bounded directed graphs, Klein and Subramanian~\cite{klein1993linear} gave a parallel work-efficient $\tilde{O}(\sqrt{n})$-depth exact SSSP algorithm. For bounded-treewidth graphs, Chaudhuri and Zaroliagis~\cite{chaudhuri1998shortest} gave a work-efficient $O(\log^2 n)$-depth algorithm for exact SSSP.


\subsection{Parallel and Distributed Algorithms for Approximate Shortest Paths}

The investigation into efficient computation of approximate distances in the parallel and distributed setting has flourished and been a tremendous source of fundamental ideas. These ideas include





\section{Preliminaries and Notation}
\label{sec:preliminaries}

Let $G$ be a directed or an undirected graph with $n := |V(G)|$ nodes and $m := |E(G)|$ edges. We use $\Nin_G(v), \Nout_G(v)$ to denote the in- and out-neighborhood of a node $v$ in $G$. For undirected graphs, $\Nin_G(v) = \Nout_G(v)$.
All our graphs are weighted. More precisely, they come with a length function $\ell_G: E(G) \mapsto \R_{\ge 0}$ which assigns each edge $e \in E(G)$ a nonnegative length $\ell_G(e)$; if $e = (u,v)$, we simplify the notation and write $\ell_G(u,v)$. However, most of our results need to assume that the input graph has only \textbf{positive integer} lengths that are polynomially bounded, albeit this property might be lost throughout the algorithm. We will make it clear in our formal statements when do we assume this is the case.
We also define the \textbf{maximum distance} $\Delta(G) := \max_{u, v \in V(G)} d_G(u, v)$, or just $\Delta$ when $G$ is clear from context.


Unless otherwise stated, all of our results hold for both directed and undirected graphs. In the latter case, we think of undirected graphs as directed graphs where each undirected edge is replaced by two opposite directed edges with the same weight.




\subsection{Boosting for directed SSSP}\label{sec:boosting-sssp}

In this section, we show how to \emph{boost} smoothly $\alpha$-approximate distance estimates to exact distances with $O(\alpha \log n)$ calls to the approximate oracle (and an insignificant amount of additional processing). In other words, finding an exact solution for SSSP on directed graphs is essentially of the same hardness as finding an (appropriate) approximate solution. Moreover, this reduction is fairly general in that it works in (at least) sequential, parallel, and even distributed models.

The core of the method is the following graph transformation: given a directed graph $G$ and a smoothly $\alpha$-approximate  distance estimate $\td$ (with respect to a source $s$), we define a new directed graph $G'$ by changing edge length of each edge $e = (u,v)$ to $\ell'(e) := \ell(e) - \td(v)/\alpha + \td(u)/\alpha$. It is easy to show that the new graph $G'$ retains nonnegative edge weights due to $\alpha$-smoothness of $\td$. Furthermore, the shortest path of $G'$ tree is also the same as the one for $G$. Various versions of this method have appeared throughout the literature~\cite{hart1968formal,johnson1977efficient,bernstein_nanogkai_wulffnilsen2022negative_sssp}.

\Cref{alg:folklore_directed_smoothing} uses the transformation to boost any approximate oracle to an exact solution: we show that repeatedly finding a smoothly $\alpha$-approximate distance estimate $\td$ via the oracle and transforming the graph using $\td$ yields a (nonnegative) graph where the shortest paths from the source $s$ to all other nodes are $0$, making the (exact) shortest path tree trivial to find. Since all of the graphs share a common shortest path tree (property of the transformation), this is sufficient to solve exact SSSP.
\begin{algorithm}[h]
	\caption{Constructing exact distances from $\alpha$-approximate ones for directed graphs}
	\label{alg:folklore_directed_smoothing}
	{\bf Input:} A directed graph $G$ with weights $\Z_{\ge 0}$ and maximum distance upper bound $\Delta$, source $s \in V(G)$, and an approximate oracle $\fO$. \\
        {\bf Oracle:} $\fO(G, s)$ returns a smoothly $\alpha$-approximate estimate $\td : V(G) \to \R_{\ge 0}$ from $s$ in $G$. \\
	{\bf Output:} Exact distances $\td' : V(G) \to \Z_{\ge 0}$ from $s$ in $G$.
	\begin{algorithmic}[1]
          \State $\ell_1 \gets \ell_G$ be the edge lengths of $G = (V, E, \ell_G)$
          \For{$i \in 1, 2, \ldots, I$, where $I \gets O(\alpha \log ( n \Delta ) )$}
          \State $\td_i \gets \fO( (V, E, \ell_i), s)$
          \State $\forall e = (u, v) \in E$,\ $\ell_{i+1}(e) \gets \ell_{i}(e) - \td_i(v)/\alpha + \td_i(u)/\alpha$
          \EndFor
          \State \Return $\td'$ defined by $\td'(v) := \sum_{i=1}^I \td_i(v) / \alpha$, rounded to the nearest integer.
	\end{algorithmic}
\end{algorithm}

\begin{lemma}[Implicit in \cite{klein1997randomized}]\label{lemma:folklore_directed_smoothing}
  Given an $n$-node directed graph $G$ with nonnegative integer weights and maximum distance $\Delta$, a source $s \in V(G)$, and access to an oracle $\fO$ which returns smoothly $\alpha$-approximate distance estimates from $s$, \Cref{alg:folklore_directed_smoothing} computes the exact distances from $s$ to each node in $G$. The oracle $\fO$ is called $O(\alpha \log (n \Delta))$ times. 
\end{lemma}
\begin{proof}
  We note that all, throughout the algorithm, edge lengths are nonnegative since $\ell'(e) = \ell(e) - \td(v)/\alpha + \td(u)/\alpha \ge 0$ can be rewritten as $\td(v) - \td(u) \le \alpha \cdot \ell(e)$, which is implied by $\alpha$-smoothness. Moreover, note that for any path from $s$ the identity $\ell_{i+1}(u, v) = \ell_i(u, v) - \td_i(v)/\alpha + \td_i(u)/\alpha$ telescopes. Formally, writing $G(i) := (V, E, \ell_i)$ and telescoping across the shortest path, we get:
\begin{align}
\label{eq:rado}
d_{G(i+1)}(s, u) = d_{G(i)}(s, u) - \td_i(u)/\alpha + \td_i(s)/\alpha = d_{G(i)}(s, u)  - \td_i(u) / \alpha. 
\end{align}
Using the fact that $\td$ is $\alpha$-approximate we have $\td_i(u) \ge d_{G(i)}(s, u)$ (for this proof, we only use noncontractivity of $\td$, with $\alpha$-smoothness this is equivalent to being $\alpha$-approximate). We argue that $\sum_{u \in V} d_{G{i}}(s, u)$ drops by a factor of $(1 - 1/\alpha)$ in each iteration:
\begin{align*}
  \sum_{u \in V} d_{G(i+1)}(s, u) = \sum_{u \in V} [ d_{G(i)}(s, u)  - \td_i(u)/\alpha ] \le (1 - \frac{1}{\alpha}) \sum_{u \in V} d_{G(i)}(s, u) .
\end{align*}
Therefore, after $I = O(\alpha \cdot \log (n \Delta))$ steps, we have $\sum_{u \in V} d_{G(I)}(s, u) < 1/3$, implying $d_{G(I)}(s, u) < 1/3$ since all distances are nonnegative. Then, by summing up \Cref{eq:rado} over $i$, we get $d_G(s, u) = d_{G(I)}(s, u) + \td'(u)$. This implies $|d_G(s, u) - \td'(u)| \le d_{G(I)}(s, u) < 1/3$, hence rounding $\td'(v)$ to the nearest integer gives the correct (integer) distance $d_G(u)$.
\end{proof}
We remark that the oracle is always called on the same graph $G$, with the exception that the lengths of the edges are changed in-between calls. Notably, these changed cause the altered graph to generally become directed (i.e., has different lengths on anti-parallel edges), even if the original graph was undirected.
Finally, we remark that in many parallel or distributed models (and many others), if the model supports an efficient implementation of $\fO$, then \Cref{alg:folklore_directed_smoothing} can also be implemented efficiently. Therefore, exact and (smoothly) approximate distances on (integer weighted) directed graphs are essentially equivalent.

\section{Different Notions of Approximate Distances}
\label{sec:distances}

This section is devoted to the exposition of all kinds of different notions of approximate distances that are either standard or are considered in this paper. We also prove fundamental results for the smooth and tree-like property and their combination. 
All of the notions discussed in this section work for directed graphs. However, we note that the notions of tree-like, strong, and to some extent also smooth distances are interesting mostly for the undirected case, since in the directed graphs we can use the boosting technique from \Cref{sec:boosting-sssp} to get even exact shortest paths.  

The roadmap for this section is as follows. 
In \Cref{sec:potentials} we first briefly review the standard notion of graph potentials. 
In \Cref{sec:def_smooth} we discuss smoothly approximate distances and explain in which sense smooth distances strenthen both approximate distances and potentials. Next, in \Cref{subsec:def_treelike} we discuss the notion \emph{tree-like} distances, i.e., those coming from an underlying approximate shortest path tree. Finally, in \Cref{sec:strong-distances} we argue that having these two properties together leads to strong distances that have a simple alternative definition.

\subsection{Graph Potentials}
\label{sec:potentials}

A notion closely related to that of approximate distances is that of graph potentials. These are distance estimates (i.e., functions $V(G) \to \R_{\ge 0}$) that are dual to approximate distances as can be formalized by the following linear program (this program will not play a role in the subsequent discussion). 

\renewcommand{\arraystretch}{1.5}
\begin{center}
  \begin{tabular}{l @{\hskip 5mm} | @{\hskip 5mm} l @{\hskip 5mm} | @{\hskip 5mm} l }    
    & \textbf{(Primal)} & \textbf{(Dual)} \\   
  \textbf{Variables:} & $f : E \to \R_{\ge 0}$ & $\phi : V \to \R$ \\
  \textbf{Optimize:} & $\displaystyle \max.\ \sum_{e \in E} \ell(e) f(e)$ & $\displaystyle \min.\ \phi(t) - \phi(s)$ \\
  \textbf{Such that:} & $\displaystyle \sum_{v \in \Nout(u)} f(u, v) - \sum_{v \in \Nin(u)} f((v,u)) =$ & $\displaystyle \phi(v) - \phi(u) \le \ell(u, v)$ \\
  & $\qquad = \1{u=s}-\1{u=t}\quad \forall u \in V$ & $\qquad \forall (u, v) \in E$
\end{tabular}
\end{center}

Formally, we define $\beta$-approximate potential as follows. 
\begin{definition}[$\beta$-approximate potential]
\label{def:potential}
A function $\phi: V(G) \mapsto \R_{\ge 0}$ is a potential with respect to $s$ if $
    \phi(s) = 0$
and
\begin{align}
  \label{eq:pot_def1}
  \forall (u,v) \in E(G): \ \ \phi(v)  \leq \phi(u) +  \ell_G(u, v).
\end{align}
A potential $\phi$ is $\beta$-approximate (with respect to a source $s$) for some $\beta \ge 1$ if
\begin{align}
  \label{eq:pot_def2}
    \beta \cdot \sum_{u \in V(G)} \phi(u) \geq  \sum_{u \in V(G)} d_G(s,u).
\end{align}
\end{definition}

Again, note that for $\beta = 1$ we recover exact shortest-path distances. 
In general, potentials always underapproximate the actual distances in the sense that $\phi(v) - \phi(u) \le d_G(u, v)$. Hence, they are $\beta$-approximate iff they do not underapproximate by more than a $\beta$ factor, on average. 

We note that the boosting result \Cref{lemma:folklore_directed_smoothing} from \Cref{sec:boosting-sssp} in fact does not require to start with smooth distance functions. It can boost any $\beta$-approximate potential to the exact shortest-path distance function in $O(\beta \log n)$ iterations.

\subsection{Smoothly Approximate Distances}
\label{sec:def_smooth}

The most important concept analysed in this paper is that of \emph{smoothly} $\alpha$-approximate distances, as defined next. 
The discussion of their applications is deferred to \Cref{sec:application}. For convenience we now repeat the definition of the $\alpha$-smooth condition that was given in \Cref{def:smooth+treelike}. 

\begin{definition}[Smooth distances] 
\label{def:smooth}
  Given a graph $G = (V, E, \ell)$ and a source $s \in V$, we say a distance estimate $\td : V \to \R_{\ge 0}$ is \emph{$\alpha$-smooth} if $
  \td(s) = 0
  $
  and
  \[
  \forall u,v \in V: \ \ \td(v) - \td(u) \leq \alpha \cdot d_G(u,v),
  \]
  Moreover, we say that a distance estimate is smoothly $\alpha$-approximate if it is both $\alpha$-smooth and noncontractive, i.e., it moreover satisfies
    \begin{align}
    \forall u \in V: \ \ d_G(s, u) \le \td(u). 
  \end{align}
\end{definition}

Observe that smoothly $\alpha$-approximate distances strengthen $\alpha$-approximate distances. 
Indeed, if we set $u = s$ in the above $\alpha$-smoothness definition, we exactly recover the bound $\td(v) \le \td(s) + \alpha \cdot d_G(s, v) = \alpha \cdot d_G(s, v)$ for $\alpha$-approximate distance estimate. 

\textbf{Subtractive triangle inequality.} The $\alpha$-smoothness requirement $\td(v) \leq \td(u) + \alpha \cdot d_G(u,v)$ can be seen as requiring an \emph{approximate triangle inequality} to hold for distances to $s$. 
In particular, for exact distances (i.e., $\alpha = 1$) the triangle inequality stipulates correctly that $d_G(s, v) \le d_G(s,u) + d_G(u,v)$. However, this inequality completely fails to hold --- even approximately --- if exact distances are replaced by approximate ones (see \Cref{fig:smooth_path}). Our strengthening demands that this inequality remains to hold $\alpha$-approximately. 
This interpretation explains most easily why smoothly approximate distances can replace exact distances in some contexts: Any use of the above triangle inequality by a proof of correctness still holds approximately if smoothly approximate distances are used to replace exact distances (e.g., see \Cref{sec:application} for an example).

\textbf{Smoothness vs. potentials.} Next, we discuss the relationship of smoothly approximate functions with potentials. 
Observe having the $\alpha$-smooth property is, up to rescaling, the same as being a potential. 

\begin{fact}\label{lemma:smooth-to-potential}
  A distance estimate $\td : V(G) \to \R_{\ge 0}$ is $\alpha$-smooth if and only if $\td / \alpha$ is a potential.
\end{fact}
\begin{proof}
Recall that potentials satisfy $\phi(v) - \phi(u) \le d_G(u,v)$ for all $u,v \in V(G)$. Then the result follows because $\td(v) - \td(u) \leq \alpha \cdot d_G(u,v)$ is equivalent to $\frac{\td(v)}{\alpha} \le \frac{\td(u)}{\alpha} + d_G(u,v)$.
\end{proof}

A helpful view of smoothly $\alpha$-approximate distances is that they get the ``best of the both worlds'': they are exactly the distances that strengthen both $\alpha$-approximate distances and $\alpha$-approximate potentials. 

\begin{lemma}
  $\td : V(G) \to \R_{\ge 0}$ is a smoothly $\alpha$-approximate distance estimate (w.r.t. $s$) if and only if
  \begin{enumerate}
        \item $\td$ is an $\alpha$-approximate distance estimate (w.r.t. $s$), and
        \item $\td/\alpha$ is an $\alpha$-approximate potential function (w.r.t. $s$). 
  \end{enumerate}
\end{lemma}
\begin{proof}
The ``only if'' direction follows from our observations that smoothly $\alpha$-approximate distances are $\alpha$-approximate and \Cref{lemma:smooth-to-potential}. 
To prove the ``if'' direction, we need to prove that $\td$ is noncontractive (\Cref{eq:noncontractivity}) and $\alpha$-smooth (\Cref{def:smooth}). The first property follows from the fact that $\td$ is $\alpha$-approximte. The second property follows from \Cref{lemma:smooth-to-potential}. 
\end{proof}

To summarize, the  smoothly approximate distances generalize approximate distances since they additionally satisfy ``subtractive triangle inequality''. But, rescaled by $\alpha$-factor, they also generalize approximate potentials, since instead of satisfying the aggregate condition $\sum_{u \in V(G)} \td(u) \ge \sum_{u \in V(G)} d_G(s,u) / \alpha$, they in fact satisfy a stronger, individual condition: for each $u$ they have $\td(u) / \alpha \ge d_G(s, u) /\alpha$. 

\subsection{Tree-Like Distances}
\label{subsec:def_treelike}

In this subsection we will discuss the $\alpha$-approximate shortest path trees --- the object that many approximate shortest path algorithms recover together with  $\alpha$-approximate distance functions. 
Formally, $\alpha$-approximate shortest path trees are defined as follows. 
\begin{definition}[$\alpha$-approximate shortest path tree]
  \label{def:tree}
  A spanning tree $T$ of a graph $G$ (where $T \subseteq G$) and source node $s$ is an $\alpha$-approximate shortest path tree for some $\alpha \geq 1$ if 
$$\forall u: \ \ d_T(s,u) \leq \alpha \cdot d_G(s,u).$$
\end{definition}
We remark that any subtree-induced distances can only overapproximate the true distances, hence $T$ being an ($\alpha$-approximate) shortest path tree automatically implies that the function $d_T(s, \cdot)$ is noncontractive. 

We find it useful to define \emph{tree-like} distance functions --- roughly speaking these are the distance functions such that we can simply recover their shortest path tree. The usefulness of this concept stems from the fact that it is sometimes easier to work with a simple algebraic condition than to work with a special tree-structure on top of the distance function. Also, we will see in \Cref{sec:strong-distances} that being tree-like strengthens being noncontractive from \Cref{eq:noncontractivity} in \Cref{def:approx-dist} similarly to how being $\alpha$-smooth strengthens the property \Cref{eq:approximation}. 
For convenience, we now repeat the definition of the tree-like property from \Cref{def:smooth+treelike}. 

\begin{definition}[Tree-Like Distances] \label{def:treelike}
  Given a graph $G = (V, E, \ell)$ and a source $s \in V$, we say a distance estimate $\td : V \to \R_{\ge 0}$ is: \emph{tree-like} (w.r.t. $s$) if $\td(s) = 0$ and $$\forall v \neq s: \ \exists e=(u,v) \in E: \ \ \td(u) \leq \td(v) - \ell(e).$$
\end{definition}

Observe that tree-like condition strengthens noncontractivity of \Cref{eq:noncontractivity} in \Cref{def:approx-dist} since from any node $v$ we can use the tree-like condition repeatedly to construct a sequence of nodes starting from $v = v_1, v_2, \dots, v_k = s$ with $\td(v_{i+1}) \le \td(v_i) - \ell(v_i, v_{i+1})$. Together with the requirement that $\td(s) = 0$, this implies $\td(v) \ge d_G(s, v)$. 

Note that above argument requires that all edge lengths are strictly positive. 
However, one can often simply transform a graph with nonnegative edge-lengths so that its edges then have only positive lengths: we can add a very small length to every edge or, in case of undirected graphs, we first find connected components of $0$-length edges and contract them to get exclusively positive lengths.

Let us now show in which sense tree-like distance functions correspond to approximate shortest path trees. 

\begin{lemma}
\label{lem:tree-like_vs_tree}
Let $G$ be a graph with edges of positive weights. 
\begin{enumerate}
    \item Let $T$ be an $\alpha$-approximate spanning tree of $G$ rooted at $s$. Then $d_T(s,\cdot)$ is a tree-like, $\alpha$-approximate distance estimate. 
    \item Let $\td$ be a tree-like distance estimate and define a tree $T$ as follows. 
    For each $v \in V(G), v \not = s$, consider an arbitrary $u \in \Nin_G(v)$ satisfying $\td(u) \le \td(v) - \ell_G(u,v)$; add $(u,v)$ to $E(T)$. 
    
    We have $d_T(s,v) \le \td(v)$ and, in particular, if $\td$ is $\alpha$-approximate then $T$ is an $\alpha$-approximate shortest path tree. 

\end{enumerate}

\end{lemma}
Note that we are again requiring positive and not nonnegative edge weights here. 
\begin{proof}
We start with the first claim. By definition, for every $v \not= s$ the predecessor $u$ of $v$ in $T$ satisfies $d_T(s, u) = d_T(s, v) - \ell_G(u, v)$ and hence $d_T(s, \cdot)$ is tree-like.
  
We continue with the second claim. First, note that $T$ is a well defined tree since we assume there are no $0$-length edges. The property $d_T(s, v) \le \td(v)$ then holds by induction on the unique path from $s$ to $v$ in $T$. Whence, whenever $\td(v) \le \alpha \cdot d_G(s, v)$ for all $v \in V(G)$, we conclude that $d_T(s, v) \le \alpha d_G(s, v)$ and $T$ is an $\alpha$-approximate shortest path tree. 
\end{proof}


Finally, let us observe that the tree-like property does not perfectly match the intuition that the function corresponds to a shortest path tree. It may seem that a better definition of tree-likeness would strengthen $\le$ to $=$ in \Cref{def:treelike}, i.e., perhaps we should require that
\begin{align}
  \forall v : \exists u \in \Nin(v) : \td(u) = \td(v) - \ell_G(u, v). \label{eq:tree_option2}
\end{align}
Above property would allow us to write $d_T(s, v) = \td(v)$ in the second part of \Cref{lem:tree-like_vs_tree}, i.e., such a function would be directly induced by $T$. However, this alternative definition does not seem to behave as nicely as our definition. For example, the property defined by \Cref{eq:tree_option2} is not closed under element-wise min-operations, unlike tree-likeness (\Cref{lem:closure}). Also, our main technical contribution, the smoothing algorithm from \Cref{sec:smooth}, preserves the tree-like property but does not preserve \Cref{eq:tree_option2}.




\subsection{Strongly approximate distances}\label{sec:strong-distances}
When a single distance estimate $\td$ has both the  $\alpha$-smooth property and is tree-like, we give it a special label of being a \emph{strongly} $\alpha$-approximate distance. 
This is to emphasize the synergy between the two properties, which we will discuss in this section. 

\begin{definition}\label{def:strong}
We say that a distance estimate is strongly $\alpha$-approximate if it has both the $\alpha$-smooth property and is tree-like. 
\end{definition}
Strong distances are especially suited for applications --- a case study is given in \Cref{sec:application}.

We argue that strong distances are ``as close as it gets'' to exact distances. In fact, we can make this claim precise: a strong distance estimate in $G$ is, in fact, the exact distance in $G$ with edge lengths perturbed by at most $\alpha$.
\begin{lemma}
\label{lem:smooth+treelike=cool}
Let $G = (V, E, \ell)$ be a graph with positive lengths, $s \in V(G)$, and $\td : V(G) \to \R_{\ge 0}$ be a distance estimate. The following two claims are equivalent.
\begin{itemize}
\item The distance estimate $\td$ is strongly $\alpha$-approximate.
\item There exist edge multipliers $1 \le \lambda(e) \le \alpha$ for every $e \in E(G)$ such that the following holds. Let $G' = (V, E, \ell')$ with $\ell'(e) := \lambda(e) \ell(e)$ for every $e \in E(G)$ be the graph $G$ with edges multiplicatively perturbed by $( \lambda(e) )_{e \in E(G)}$. Then, $\td$ is an \emph{exact} distance function on $G'$ in the sense that $\td(u) = d_{G'}(s, u)$ for all $u \in V(G)$.
\end{itemize}
\end{lemma}
\begin{proof}
  We prove the result for directed graphs, it is simple to adapt it to undirected graphs. We first prove that the first claim implies the second. We are going to define the multipliers in such a way that the tree given by the tree-like property of $\td$ becomes the exact shortest path tree in $G'$. 
  
  More concretely, we define the multipliers $\lambda$ as follows. First, we use that $\td$ is tree-like and let every node $v$ choose its neighbor $u$ such that $\td(u) \le \td(v) - \ell_G(e)$ where $e = (u,v)$.  
  The smoothness of $\td$ implies that $\td(u) \ge \td(v) - \alpha \ell_G(e)$. 
  Putting the two properties together, we infer that 
  \[
  \ell_G(e) \le \td(v) - \td(u) \le \alpha\ell_G(e). 
  \]
  Thus we can choose $\lambda(e)$ such that $\ell_{G'}(e) = \lambda(e) \ell_G(e) = \td(v) - \td(u)$. 

    For all other edges $e = (u,v)$ we choose $\lambda(e) = \alpha$; note that this implies $$\td(v) - \td(u) \le \alpha \ell_G(e) = \ell_{G'}(e).$$ 

  We claim $\td$ are exact distances on $G'$ which we show by proving that $\d_{G'}(s, v) \le \td(v) \le \d_{G'}(s, v) $. To observe that $\td(v) \le d_{G'}(s, v)$ for all $v \in V(G)$, note that by construction we have that $\td(v) - \td(u) \le \ell_{G'}(e)$. That is, $\td$ is a potential function on $G'$ and in particular $\td(v) \le d_{G'}(s,v)$ for all $v\in V(G)$. 
  
    On the other hand, for every $v$ there is $u$ with $\td(u) < \td(v)$ (as all edges have nonzero lengths) where by construction we have $\td(v) = \td(u) + \ell_{G'}(u,v)$. 
  Hence, we can inductively find a path from $s$ to $v$ such that $\td(v) = d_{P \cap G'}(s, v)$. Consequently, we infer that $\td(v) \ge d_{G'}(s, v)$ (i.e., is noncontractive on $G'$) and this completes the proof of the first part of the statement. 

 It remains to prove that the second claim implies the first one in the lemma statement. Since $\td$ is exact on $G'$, for an edge $(u, v) \in E(G)$ we have $\td(v) - \td(u) \le \ell'(u, v) \le \alpha \cdot \ell(u, v)$. Hence, $\td$ has the $\alpha$-smooth property. 
 Furthermore, let $T'$ be the (exact) shortest path tree from the source $s$ in $G'$. 
 For each node $v$ let $(u, v) \in E(G)$ be the parent edge of $v$ in $T'$. By definition, we have $\td(u) = \td(v) - \ell'(u, v) \ge \td(v) - \ell(u, v)$ since $\ell(u,v) \le \ell'(u, v)$, thus $\td$ is tree-like w.r.t. $s$. Therefore, $\td$ is strongly $\alpha$-approximate, as needed.
\end{proof}


\textbf{Closure properties.} We observe that the discussed distance notions have useful closure properties, which will be useful later.

\begin{lemma}
\label{lem:closure}
Let $\td_1, \td_2$ be two distance estimates. 
Let us define $\td_{\min} = \min(\td_1, \td_2)$, $\td_{\max} = \max(\td_1, \td_2)$, $\td_{\avg} = \lambda \td_1 + (1-\lambda) \td_2$ for some $0 \le \lambda \le 1$ (element-wise). We have:
\begin{enumerate}
    \item If both $\td_1, \td_2$ are $\alpha$-approximate, so are $\td_{\min}, \td_{\max}, \td_{\avg}$. 
    \item If both $\td_1, \td_2$ have the $\alpha$-smooth property so does $\td_{\min}, \td_{\max}, \td_{\avg}$.
    \item If both $\td_1, \td_2$ are tree-like, so is $\td_{\min}$.
\end{enumerate}
\end{lemma}
\begin{proof}
The item (1) follows from the chain of inequalities $d_G(s, v) \le \td_{\min}(v) \le \td_{\avg}(v) \le \td_{\max}(v) \le \alpha d_G(s, v)$. 

We continue with item (2). We start with $\td_{\min}$, let $u,v$ be any two nodes of $G$ and without loss of generality assume that $\td_{\min}(u) = \td_1(u)$. We then have $\td_{\min}(v) - \td_{\min}(u) \le \td_1(v) - \td_1(u) \le \alpha d_G(u, v)$ as needed. 
For $\td_{\max}$, let us similarly without loss of generality assume $\td_{\max}(v) = \td_1(u)$. 
We then have $\td_{\max}(v) - \td_{\max}(u) \le \td_1(v) - \td_1(u) \le \alpha d_G(u,v)$ as needed. Finally, for $\td_{\avg}$ we can write $$\td_{\avg}(v) - \td_{\avg}(u) = \lambda\left( \td_1(v) - \td_1(u) \right) + (1 - \lambda)\left( \td_2(v) - \td_2(u) \right) \le (\lambda + (1-\lambda)) \cdot \alpha d_G(u,v)$$ and we are done. 

We finish with item (3). Fix a node $v$ and without loss of generality assume $\td_{\min}(v) = \td_1(v)$. Consider the node $u$ such that $\td_1(u) \le \td_1(v) - d_G(u, v)$. We have $\td_{\min}(u) \le \td_1(u) \le \td_1(v) - d_G(u,v) = \td_{\min}(v) - d_G(u,v)$ as needed. 
\end{proof}
We note that $\td_{\avg}, \td_{\max}$ are not necessarily tree-like even if $\td_1, \td_2$ are. This can happen e.g. whenever $v$ has two predecessors $u_1, u_2$ such that $\td_1(u_1) \le \td_1(v) - d_G(u,v)$ but $\td_1(u_1) > \td_1(v) - d_G(u,v)$ and analogously for $\td_2$.

\textbf{Local checkability.} 
What do $\alpha$-smoothness and tree-likeness have in common? We observe that $\alpha$-smoothness is a locally checkable proof that a distance function satisfies the property \Cref{eq:approximation}, whereas tree-likeness is a locally checkable proof that a distance function satisfies the noncontractiveness property \Cref{eq:noncontractivity}. Therefore, strongly $\alpha$-approximate distances are $\alpha$-approximate distances where both required properties can be checked locally. 

We briefly discuss the topic of local checkability: Given a distance estimate $\td$ that is presumably $\alpha$-approximate, it is often convenient to be able to prove its approximation guarantee without trusting the algorithm; or, e.g., see whether it still holds after a graph changes to avoid re-running the computation. One way to make this certification precise is to look for local algorithms that verify the property (see \cite{Feuilloley2019local_certification_survey} for an introduction to the topic of local certification). More precisely, each node is given the values in its (small-hop) neighborhood, and outputs YES/NO; the property holds if and only if all nodes output YES.

It is not possible to directly check whether $\td$ is $\alpha$-approximate without computing the exact distances first. 
On the other hand, one can easily check whether a distance estimate $\td$ is strongly $\alpha$-approximate via a local algorithm. 
Concretely, for the $\alpha$-smoothness property, each node $u$ checks for all outgoing edges $(u,v)$ whether $\td(v) \le \td(u) + \alpha d_G(u,v)$ (outputs NO if any fail). 
This works since requiring that the $\alpha$-smooth condition holds for all edges $(u,v)$ already implies it holds for any pair $u,v$. 
Tree-likeness is similar; together, we conclude that one can locally verify whether $\td$ is strongly $\alpha$-approximate. 

The fact that a property can be easily checked is often helpful. For example, it allows us to turn any randomized Monte Carlo algorithm for strongly approximate distance computation into a Las Vegas one. 

\subsection{Do Known Algorithms give Strong Distance Guarantees?}\label{sec:nobody-gets-strong}

Many important algorithms use exact shortest path as an algorithmic building block (e.g., low-diameter decompositions, $\ell_1$-embeddings, hopsets, emulators, various versions of min-cost maximum flows, etc.). However, obtaining exact distances is hard in many settings (e.g., parallel or distributed), and replacing exact distances with an arbitrary $(1+\eps)$-approximation is insufficient. For this reason, there is a long list of papers which obtain approximate distances with additional guarantees. However, as far as the authors as aware, only a few approaches obtain strong distances, and they only do it in very specific settings that cannot be readily generalized. In this section, we review several algorithms from the parallel and distributed literature and clarify their properties.


\textbf{Hopsets.} A hopset with hopbound $\beta \ge 1$ of a graph $G = (V, E, \ell)$ is a (sparse) graph $H = (V, E', \ell')$ such that $\beta$-hop paths in $G \cup H$ $(1+\eps)$-approximate shortest path of $G$. More precisely, $d_G(u, v) \le d_{G \cup H}^{(\beta)}(u, v) \le (1 + \eps) d_G(u, v)$ for all $u, v \in V$, where $d^{(\beta)}(u, v)$ is the shortest path using paths of at most $\beta$ hops~\cite{cohen2000polylog,elkin2019hopsets,huang2019thorup,cao2020paralleldirected}. Simple $O(\beta)$-time algorithms like Bellman-Ford can compute $\td(v) := d_{G \cup H}^{(\beta)}(s, v)$, which $(1+\eps)$-approximate SSSP distances from $s$. Moreover, every edge in $H$ corresponds to a path in $G$, hence it is not hard to prove that such $\td$ is tree-like. On the other hand, $\td$ is generally not smooth. Hence, hopsets do not directly give strong distances.

\textbf{Algorithms computing approximate potentials.} Computing approximate potentials (i.e., smooth and approximate distance) can be used to boost the approximation guarantees for distances in both directed~\cite{klein1997randomized} and undirected settings~\cite{sherman2017generalized,zuzic2021simple}. Therefore, many papers develop tools to compute approximate potentials. For example, the state-of-the-art $(1+\eps)$-approximate undirected shortest path papers in parallel and distributed setting~\cite{Li20shortest_paths,becker2021near,rozhon_grunau_haeupler_zuzic_li2022deterministic_sssp,goranci2022universally} all can compute a primal-dual pair $(f, \phi)$ where $f$ is a $(1+\eps)$-approximate shortest path tree, and $\phi$ is a $(1+\eps)$-approximate potential. However, these approaches do not yield a single distance estimate $\td$ that is strong (i.e., has primal-dual guarantees)---the primal tree is not smooth, and the potentials are not tree-like.

\textbf{Emulator approach.} Several papers do obtain strong approximate distances using emulators, which are graph $H$ such that $d_H(u, v) \le d_G(u, v) \le \alpha \cdot d_H(u, v)$ for all $u, v$ (i.e., an emulator is an $\alpha$-spanner without the subgraph requirement). These papers construct an emulator with additional structural properties in a way that it is possible to compute \emph{exact} distances on them. These are strong distances per \Cref{lem:smooth+treelike=cool}. For example, Andoni, Stein, and Zhong~\cite{andoni_stein_zhong2020shortest_paths} compute a $\poly(\log n)$-emulator where $O(\log \log n)$-hop paths exactly match the shortest paths, hence can be computed in their parallel setting. In another approach, Forster and Nanongkai~\cite{forster2018faster} (and several other papers that follow this line of work) randomly sample about $\sqrt{n}$ nodes, compute a $2$-emulator on it, and then compute exact distances on this emulator (which can be extended to all nodes using small number of hops). Unfortunately, it is unclear how to make these approaches work outside of their respective settings.

\section{The Smoothing Algorithm}
\label{sec:smooth}
This section is devoted to the proof of our main techincal result, \Cref{thm:smooth} that we restate here for convenience. 

\begin{restatable}{theorem}{smooth}{Obtaining smoothness}
\label{thm:smooth}
Given a graph $G$ with (real) weights in $[1, \poly(n)]$, a source $s \in V(G)$, accuracy $\eps \in (0,1]$, and an approximate oracle $\fO$, \Cref{alg:smoothing} computes $(1+\eps)$-approximate distances that are $(1+\eps)$-smooth $\td$ in a directed graph $G$ using $O(\log n)$ calls to a $(1+O(\eps/\log n))$-approximate distance oracle $\fO$ on directed graphs. Moreover, if $\fO$ returns tree-like distances, the algorithm returns strongly $(1+\eps)$-approximate distances. Moreover, if $G$ is undirected, it suffices if $\fO$ works in undirected graphs.  
\end{restatable}

First, in \Cref{sec:smooth_intuition} we build the intuition behind the algorithm. Next, in \Cref{sec:directed_algorithm} we prove \Cref{thm:smooth} formally. 

\subsection{Intuition Behind the Algorithm}
\label{sec:smooth_intuition}

Let us now explain the intuition behind our algorithm that uses calls to $(1+\eps)$-approximate distance oracle to compute smooth approximate distances in directed graphs. 
We will build the algorithm in three steps.

\begin{figure}
    \centering
    \includegraphics[width = \textwidth]{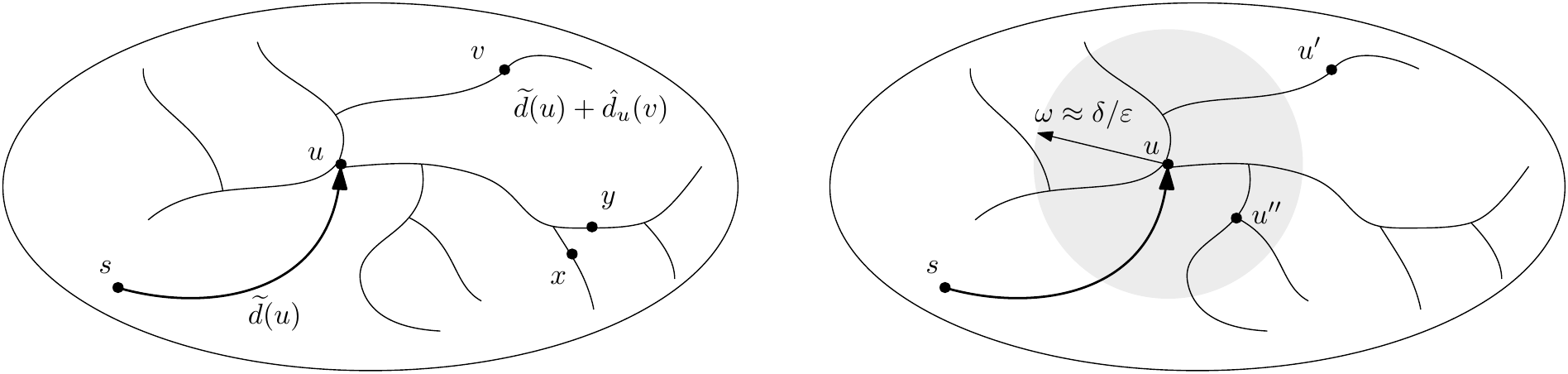}
    \caption{Left: In the first try, we simply try to ``fix'' a node $u$ (by that we mean making all pairs $u,v$ satisfy the smoothness condition) by computing approximate distances $\hd_u$ from $u$. Adding $\td(u)$ to all computed values then creates a distance estimate $\td_u$ measuring distances from $s$. This distance estimate fixes $u$; even if we set $\td' = \min(\td, \td_u)$, $u$ remains fixed in $\td'$. The problem is that even though a pair $x,y$ may have satisfied the smoothness condition in $\td$, we may unfix it in $\td'$. \\ 
    Right: To deal with the problem,  we use the same construction but slow down the edges of $G$ by a factor of $(1+\eps)\alpha$. We also do not aim to make $u$ satisfy the $\alpha$-smoothness condition right away but instead only go from $(\alpha, \delta)$-smoothness to $((1+\eps)\alpha, \delta/2$-smoothness guarantee. Because of the slowdown and assumption that $\td$ is $(\alpha, \delta)$-smooth, any node $u'$ that is further than $\omega = 2\delta/(\alpha\eps) = \Theta(\delta/\eps)$ from $u$ has $\td(u') \ge \td_u(u')$ which means that $\td'(u) = \td(u)$ and we do not unfix it. On the other hand, whenever a node $u''$ is closer than $\omega = \Theta(\delta/\eps)$ to $u$, we use the fact that the approximation oracle can make the value $\hd(u'')$ distorted from the truth only additively by $\frac{\eps}{100} \cdot \omega \ll \delta/2$. Hence, up to a small error dominated by $\delta/2$ we can use the same argument that shows that  $u$ gets fixed for $u''$.  }
    \label{fig:intuition}
\end{figure}

\paragraph{Step 1: The first try.}

Recall that we wish to compute the $(1+\eps)$-approximate smooth distances from some source node $s$. We can start by  running the approximate distance oracle from $s$ to get some approximate distance estimate $\td$. 
The problem of course is that there can be two nodes $u,v$ ($v$ further from $s$ than $u$) such that $\td(v) - \td(u)$ is much larger then $(1+\eps) d_G(u,v)$. 
This is because our oracle is allowed to make an additive error of $\eps \cdot d_G(s, u)$ for the distance labels $\td(u), \td(v)$; this can be much larger than the distance $d_G(u,v)$. 

Let us first discuss how we can use one more call to the approximate distance oracle to satisfy the smoothness condition for some fixed pair of nodes $u,v$:
We run the approximate distance oracle for the second time, but we start at $u$. We get an approximate distance estimate $\hd_u$ from $u$ that defines a new noncontractive distance estimate $\td_u$ from $s$ by defining $\td_u(v) = \td(u) + \hd_u(v)$ (see \Cref{fig:intuition}). 

This new distance estimate $\td_u$ is not necessarily $(1+\eps)$-approximate, as it considers only routes through $u$. 
However, the pair of nodes $u$ and $v$ (for any $v$) satisfies the smoothness condition: we have  $\td_u(v) - \td_u(u) = (\td(u) + \hd(v)) - \td(u)  = \hd(v) \le (1+\eps) d_G(u,v)$. 

Finally, if we define $\td'(\cdot) = \min(\td(\cdot), \td_u(\cdot))$, the new estimate takes the best of the both worlds: $\td'$ remains a noncontractive $(1+\eps)$-approximate distance estimate that also additionally satisfies the smoothness condition for $u$ and any other node $v$ (to verify this, compare with the proof of \Cref{lem:closure}).

So what is the bad news? It may seem that we can already construct a simple, albeit too slow, smoothing algorithm that simply repeats above construction for all nodes $u \in V(G)$, until all pairs of nodes satisfy the smoothness condition. However, such an algorithm does not work! 
While the above procedure for some $u$ ``fixes'' all the pairs $(u,v)$, it may also ``break'' some other pair $(x, y)$ elsewhere in the graph that satisfied the smooth condition in $\td$.


\paragraph{Step 2: A (formally) correct but slow algorithm.}
We will now show how we can fight the problem that fixing one pair $u,v$ by computing distances from $u$ can break some other pair of nodes in the graph. Our algorithm will still require $O(n \log n)$ calls to the distance oracle which we improve upon only in the final step. 

Our main idea is, roughly, to fix the pairs $u,v$ in $O(\log n)$ phases, starting from the pairs with the largest $d_G(u,v)$. 
To be more concrete, let us generalize the smoothness condition as follows:

\begin{definition}[$(\alpha, \delta)$-smoothness]
We say that a distance estimate $\td$ is $(\alpha, \delta)$-smooth if for every $u,v \in V(G)$ we have 
\[
\td(v) \le \td(u) + \alpha \cdot d_G(u,v) + \delta
\]
\end{definition}
We will show how to refine an $(\alpha, \delta)$-smooth distance estimate $\td$ to a $((1+\eps)\alpha, \delta/2)$-smooth estimate $\td_*$.
Then, to get $(1+\eps_0)$-smoothness, we simply choose $\eps = \Theta(\eps_0 / \log n)$, start with the trivial $(1, \poly(n))$-smooth estimate (the weights are polynomially-bounded integers) and use the above reduction $O(\log n)$ times until we end up with $(1+\frac{\eps_0}{2}, \frac{1}{\poly(n)})$-smooth distances that are also $(1+\eps_0)$-smooth. On a more intuitive level, we first fix pairs at larger distance scales before moving on to smaller distance scales. Indeed, $(\alpha, \delta)$-smooth distances $\td$ are $(\alpha + \eps)$-smooth for pairs $(u, v)$ with $d_G \ge \delta / \eps$:
\[
\td(v) \le \td(u) + \alpha \cdot d_G(u,v) + \delta 
\le \td(u) + (\alpha + \eps)d_G(u,v). 
\]

The rest of this explanation focuses on a single refinement step. Our algorithm is very similar to the one discussed in the previous subsection: We simply iterate over each node $u$ of $G$ and run our $(1+\eps/10)$-approximate oracle from $u$, compute the distance estimates $\td_u(\cdot)$ as before, and output the element-wise minimum. The important difference is that we run our oracle on the ``slowed down'' graph $G'$ where each edge length is multiplied by $(1+\eps/2)\alpha$. The method is formally written in \Cref{alg:smoothing_lemma_slow}.

Intuitively, the benefit of the slowdown is that node $u$ will not affect nodes far away from it, thereby not breaking them.
\begin{lemma}\label{lemma:leave-far}
  If $d_G(u, u') > 2\delta/(\alpha\eps)$, then $\td_u(u') > \td(u')$.
\end{lemma}
\begin{proof}
  By assumption of $(\alpha, \delta)$-smoothness we have $\td(u') \le \td(u) + \alpha d_G(u,u') + \delta$. Noncontractivity of distance estimates on the slowed-down graph implies $\td_u(u') \ge \td(u) + (1+\eps/2)\alpha d_G(u,u')$. Comparing the two expressions and rewriting our assumption as $(\eps/2) \cdot \alpha d_G(u,u') > \delta$ gives $\td_u(u') > \td(u')$, as required.
\end{proof}

Furthermore, $\td_u(\cdot)$ ``fixes'' the smoothness for all nodes close to $u$ (a feature of our relaxed $(\alpha, \delta)$ definition of smoothness). 

\begin{lemma}\label{lemma:fix-close}
  If $d_G(u,u') \le 2\delta/(\alpha\eps)$, then all pairs $(u', v)$ are $(\alpha(1+\eps), \delta/2)$-smooth in $\td_u(\cdot)$.
\end{lemma}
\begin{proof}
  Indeed, for any node $v$ we can write 
  \[
    \td_u(v) - \td_u(u') \le \alpha(1+\eps/2)(1+\eps/10) d_G(u,v) - \alpha(1+\eps/2) d_G(u,u')
  \]
  which, using triangle inequality $d_G(u,v) - d_G(u,u') \le d_G(u', v)$, simplifies to
  \[
    \td_u(v) - \td_u(u') \le \alpha(1+\eps/2)(1+\eps/10) d_G(u',v) + \alpha(1+\eps/2)\frac{\eps}{10} \cdot d_G(u,u')
  \]
  The last error term can now be upper bounded by $\alpha(1+\eps/2)\frac{\eps}{10} \cdot 2\delta/(\alpha\eps) \le \delta/2 $, so we get
  \[
    \td_u(v) - \td_u(u') \le \alpha(1+\eps) d_G(u',v) + \delta/2. 
  \]
  In other words, the node $u'$ now satisfies the $((1+\eps)\alpha, \delta)$-smooth condition! 
\end{proof}

We now formally prove the correctness of \Cref{alg:smoothing_lemma_slow}.
\begin{lemma}
  Every pair $(u', v)$ is $(\alpha(1+\eps), \delta/2)$-smooth in $\td_*$. 
\end{lemma}
\begin{proof}
Let $u := \arg\min_{u} \td_{u}(u')$ be the ``parent'' of $u'$. Hence, $\td_u(u') = \td_*(u') \le \td_{u'}(u') = \td(u')$. By \Cref{lemma:leave-far}, we have $d_G(u, u') \le 2\delta/(\alpha\eps)$. Hence, by \Cref{lemma:fix-close}, $(u', v)$ is $((1+\eps)\alpha, \delta)$-smooth in $\td_u$. But, $\td_*(v) - \td_*(u') = \td_*(v) - \td_u(u') \le \td_u(v) - \td_u(u')$, hence $(u', v)$ is also smooth in $\td_*$. This completes the proof.
\end{proof}

\begin{algorithm}[ht]
	\caption{Slow Partial Smoothing algorithm}
	\label{alg:smoothing_lemma_slow}
	{\bf Input:} A directed graph $G$, source $s$, parameter $\eps$, approximate oracle $\fO$, $(\alpha,\delta)$-smooth distance estimate $\td$\\
        {\bf Oracle:} $\fO(G, s, \eps)$ returns $(1+\eps)$-approximate distances $\td : V(G) \to \R_{\ge 0}$ from $s$ in $G$. \\        
	{\bf Output:} $((1+\eps)\alpha, \delta /2)$-smooth distance estimate $\td'$
	\begin{algorithmic}[1]
          \State Let $G'$ be the graph $G$ with distances multiplied by $(1+\eps/2)\alpha$.
          \State $\td \gets \fO(G, s, \eps)$
            \For{each $u \in V(G)$}
                \State $\hd_{u} \leftarrow \fO(G', u, \eps/10)$
                \State $\td_{u}(\cdot ) \leftarrow \td(u) + \hd_{u}(\cdot )$
            \EndFor
            \State \Return $\td_*(\cdot) = \min_{u \in V(G)} \td_{u}(\cdot)$
	\end{algorithmic}
\end{algorithm}

\paragraph{Step 3: Speeding up the algorithm.}

It remains to show how we can speed up \Cref{alg:smoothing_lemma_slow} so that it uses only $2$ queries to the distance oracle per refinement step, instead of $n$. 
To this end, we will use two observations about how we can ``bundle'' the oracle calls for two different vertices into just one call. 

First, consider some two nodes $u_1, u_2$ such that $|\td(u_1) - \td(u_2)| \le \delta/\eps$. Instead of running the distance oracle once from $u_1$ and once from $u_2$, we can use the following trick. Assuming $\td(u_1) < \td(u_2)$, construct a new special node $s$, connect it to $u_1$ with an edge of length $0$ and with $u_2$ with an edge of length $\td(u_2) - \td(u_1)$. Run just one oracle call from $s$ to get a distance estimate $\hd_s(\cdot)$; define $\td_{u_1, u_2}(v) = \td(u_1) + \hd_s(v)$. The distance function $\td_{u_1, u_2}(\cdot)$ is basically the same as the function $\min(\td_{u_1}(\cdot), \td_{u_2}(\cdot))$ that we want to compute. The difference is that the edge from $s$ to $u_2$ has length $\td(u_2) - \td(u_1) < \delta/\eps$ and this worsens the argument from previous analysis of the slow smoothing algorithm by an additive $\frac{\eps}{10} \cdot \frac{\delta}{\eps} = \frac{\delta}{10} \ll \frac{\delta}{2}$. That is, this additional loss is insignificant. 

Second, consider two nodes $u_1, u_2$ with $|\td(u_1) - \td(u_2)| \gg \delta/\eps$. Again, add a special node $s$, connect to $u_1$ and $u_2$ with an edge of length $0$ and run our distance oracle from $s$. We already know from \Cref{lemma:leave-far} that the distance information we get by calling our distance oracle from $u_1$ is valuable only for nodes that are at distance at most $O(\delta/\eps)$ from $u_1$, and analogously for $u_2$. This again suggests that the information we gain by calling our oracle from $s$ is sufficient to recover the value of $\min(\td_{u_1}(\cdot), \td_{u_2}(\cdot))$. 

We put these two observations together as follows as  visualized in \Cref{fig:level_graph}. We define two \emph{level graphs} $H_1, H_2$; to get the level graph $H_1$, we start with $G$ and slow down all edges by a factor of $(1+\eps)\alpha$. 
Next, we partition the nodes of $G$ into ``levels'' of length $\omega = \Theta(\delta/\eps)$ based on their value of $\td$; we delete edges going between different levels.  
Finally, we add a new node $\sigma$ and connect it to every node $u$ with an edge of length $\{\td(u)\}_1$ which is defined as the unique number between $0$ and $\omega$ such that $\td(u) - \{\td(u)\}_1$ is divisible by $\omega $. 

Running the distance oracle on $H_1$ enables us to (approximately) compute $\min(d_{u_1}(\cdot), d_{u_2}(\cdot), \dots)$ where the nodes $u_1, u_2, \dots$ are all the nodes of $G$ except those that are too close to the boundary of two different levels. We next define a graph $H_2$ analogously to $H_1$ but with its levels shifted by $\omega /2$. Computing the distances in $H_1$ and $H_2$ and taking the minimum of appropriate quantities then allows us to compute the value of $\min_{u \in V(G)} d_u(\cdot)$, thereby solving the problem we need to solve in \Cref{alg:smoothing_lemma_slow}. However, we need just two calls to the distance oracle instead of $n$ calls. 
Hence, the final smoothing algorithm only needs $O(\log n)$ oracle calls.

\subsection{The Formal Analysis}
\label{sec:directed_algorithm}

In this section we formally describe and analyze the smoothing algorithm described in \Cref{alg:smoothing}. Specifically, we prove its properties, restated below for convenience.

\smooth*

We start by describing the level graphs from \Cref{fig:level_graph}. 

\todo{V: the triangles in the picture are not matching the levels in the graph}
\begin{figure}
    \centering
    \includegraphics[width = .8\textwidth ]{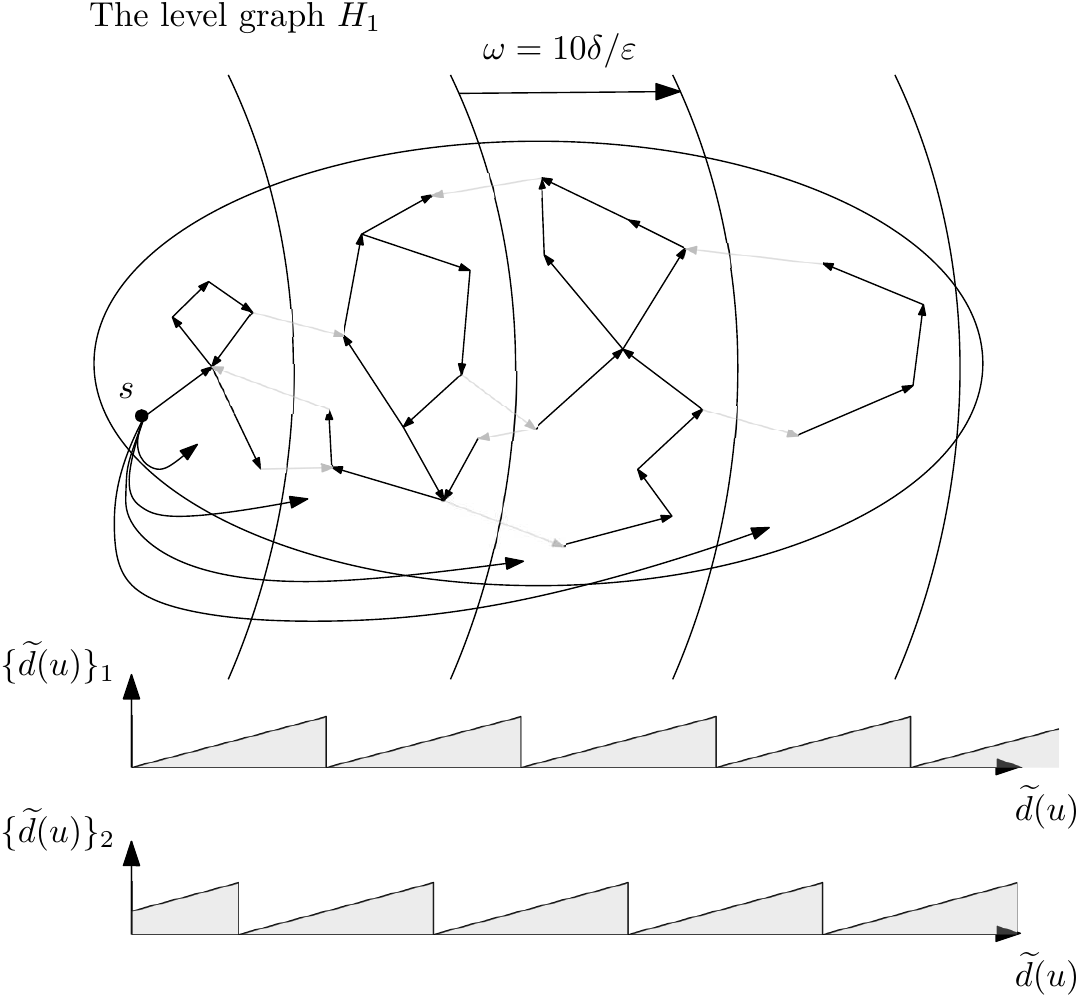}
    \caption{The figure shows one level graph $H_1$. The level graph is created from $G$ by cutting it based on $\td$ into levels of length $\omega = 10\delta/\eps$. We remove edges that go across boundaries (grey edges). We also add edges from $s$ to all other nodes $v$ of length $\{\td(v)\}_1$. This is the ``fractional part'' of $\td(v)$ as shown below the graph. The graph $H_2$ is defined analogously but the fractional parts are shifted by $\omega/2$.  }
    \label{fig:level_graph}
\end{figure}

\begin{definition}[Rounding distances]
\label{def:rounding}
Let $G$ be a graph and $\td$ a distance estimate on it. Let $i \in \{1, 2\}$ and $\omega$ a parameter. 
We define $\{\td(u)\}_i$ to be the smallest value such that $\td(u) - \{\td(u)\}_i - (i-1)\cdot \frac{\omega}{2}$ is divisible by $\omega$. 
We also define $\td(u) = \lfloor \td(u)\rfloor_i + \{\td(u)\}_i$. 
\end{definition}

\begin{definition}[Level graph $H_i$]
	\label{def:level_graph}
Given an input directed graph $G$, a width parameter $\omega$ and $i \in \{1, 2\}$, we define the \emph{level graph} $H_i$ as follows. 

We define $V(H_i) = V(G) \cup \{\sigma\}$. The set of edges $E(H_i)$ contains an edge $(u,v) \in E(G)$ if and only if $\lfloor \td(u)\rfloor_i = \lfloor \td(v) \rfloor_i$. 
The weight of the edge $(u,v)$ in $H_i$ is defined as  $$\ell_{H_i}(u,v) = (1+\eps)\alpha \ell_{G}(u, v). $$
Moreover, for every $u \in V(G)$ we add an edge $(\sigma, u)$ to $E(H_i)$ of weight $\{\td(u)\}_i$. 
\end{definition}

\begin{algorithm}[ht]
	\caption{Smoothing algorithm $\textsc{Smoothing}(G, s, \eps)$}
	\label{alg:smoothing}
	{\bf Input:} A directed graph $G$, source $s$, parameter $\eps$, approximate oracle $\fO(G, s, \eps')$ for $\eps' = O(\eps / \log n)$\\
        {\bf Oracle:} $\fO(G, s, \eps)$ returns $(1+\eps)$-approximate distances $\td : V(G) \to \R_{\ge 0}$ from $s$ in $G$. \\
	{\bf Output:} Smoothly $(1+\eps_0)$-approximate distance estimate $\td$.
	\begin{algorithmic}[1]
	    \State Let $\Delta = n^{O(1)}$ be the maximum distance in $G$. 
        \State $t \leftarrow 1 + \log_2 (n\Delta/\eps_0)$
	    \State $\eps \leftarrow \eps_0 / (10 t)$
	    \State Define distance $\td_0(u) \leftarrow \fO(G, s, 1)$  // this estimate is $(1, n\Delta)$-smooth. 
        \For{$i \leftarrow 1, \dots, t$}
            \State $\td_i(u) \leftarrow \textsc{PartialSmoothing}(\td_{i-1}(u), \eps)$ // this estimate is $\left( (1+\eps)^i, n\Delta/2^i\right)$-smooth
        \EndFor
        \State \Return $\td' \leftarrow \td_t$
	\end{algorithmic}
\end{algorithm}

\begin{algorithm}[ht]
	\caption{Partial Smoothing algorithm $\textsc{PartialSmoothing}$}
	\label{alg:smoothing_lemma}
	{\bf Input:} A directed graph $G$, source $s$, parameter $\eps$, approximate oracle $\fO$, $(\alpha,\delta)$-smooth distance estimate $\td$.\\
        {\bf Oracle:} $\fO(G, s, \eps)$ returns $(1+\eps)$-approximate distances $\td : V(G) \to \R_{\ge 0}$ from $s$ in $G$.\\
	{\bf Output:} $((1+\eps)\alpha, \delta /2)$-smooth distance estimate $\td'$
	\begin{algorithmic}[1]
            \State Define $H_1, H_2$ as in \Cref{def:level_graph} for $\omega = 10\delta/\eps$
            \For{$i \leftarrow 1, 2$}
                \State $\hd_i \leftarrow \fO(H_i, s, \eps/100)$
                \State $\td_i \leftarrow \lfloor \td(u) \rfloor_i + \hd_i(u)$
            \EndFor
            \State \Return $\td'$ defined by $\td'(\cdot ) = \min\left(\td(\cdot), \td_1(\cdot ), \td_2(\cdot ) \right)$
	\end{algorithmic}
\end{algorithm}

To prove \Cref{thm:smooth}, we need mainly to analyse the smoothing subroutine of \Cref{alg:smoothing_lemma}. To analyse it, we first define the concept of a smoothing path.

\begin{definition}[Smoothing path]
Given a parameter $\varepsilon$ and an $(\alpha, \delta)$-smooth distance estimate $\tilde{d}$ as in \Cref{alg:smoothing_lemma}, we will call a path $P$ from $u$ to $v$ in $G$ smoothing if $$\tilde{d}(v) \geq \tilde{d}(u) + \alpha(1+\varepsilon)\ell_{G}(P).$$
\end{definition}

We next prove that any smoothing path cannot be very long. 

\begin{claim}
\label{cl:short_smoothing_paths}
Any smoothing path $P$ satisfies $\ell_G(P) \leq \delta/(\alpha\eps).$
\end{claim}
\begin{proof}
We have by definition that a smoothing path $P$ from $u$ to $v$ satisfies
$$\alpha(1+\eps)\ell_G(P) \leq \tilde{d}(v)-\tilde{d}(u),$$
where by assumption on $\tilde{d}(\cdot)$ we have
$$ \tilde{d}(v)-\tilde{d}(u) \leq \alpha\cdot d_G(u, v) + \delta \leq \alpha \cdot \ell_G(P) + \delta.$$
Combining these gives the desired inequality.
\end{proof}

The fact that smoothing paths cannot be long then implies that any smoothing path is contained in either $H_1$ or $H_2$. 

\begin{claim}
\label{cl:contained}
Any smoothing path $P\subseteq G$ from $u$ to $v$ satisfies either $P \subseteq H_1$ or $P \subseteq H_2$. 
\end{claim}
\begin{proof}
Let $w_{\max} = \argmax_{w\in P}\,\tilde{d}(w)$  and $w_{\min} = \argmin_{w\in P}\,\tilde{d}(w)$. We upper bound the difference $\td(w_{\max}) - \td(w_{\min})$. We have
\begin{align*}
    \tilde{d}(w_{\max})-\tilde{d}(w_{\min}) &\leq \tilde{d}(w_{\max})-\tilde{d}(w_{\min}) + \tilde{d}(v)-\tilde{d}(u) \\
    & = \left(\tilde{d}(w_{\max})-\tilde{d}(u)\right) + \left(\tilde{d}(v)-\tilde{d}(w_{\min})\right)\\
    & \leq \alpha \ell_G(P) + \delta + \alpha \ell_G(P) + \delta \\
    & \leq 2(\eps^{-1}+1)\delta \le 3\delta/\eps
\end{align*} 
where we first used the $(\alpha, \delta)$-smoothness assumption on $(u, w_{\max})$ and $(w_{\min}, v)$ and then we applied \Cref{cl:short_smoothing_paths}. 
Note that whenever $P \not\subseteq H_1$, we have necessarily $\{\td(w_{\min})\}_1 < \omega$ while $\{\td(w_{\max})\}_1 \ge \omega$. However, our bound then implies that $\frac{\omega}{2} \le \omega - 3\delta/\eps \le  \{\td(w_{\min})\}_1 \le \{\td(w_{\max})\}_1 \le \omega + 3\delta/\eps \le \frac{3\omega}{2}$ which implies that $P \subseteq H_2$. 
\end{proof}

\begin{proposition}\label{prop:sandwich} The output distance estimate $\td'$  in \Cref{alg:smoothing_lemma}  satisfies
\begin{equation}\label{eq:smoothformula}  
\min_{u \in V(G)} \left( \tilde{d}(u) + \alpha(1+\eps) d_G(u, v) \right) 
\leq \tilde{d}'(v) 
\leq  \min_{u \in V(G)} \left( \tilde{d}(u) + \alpha(1+\eps) d_G(u, v) \right) + \delta/2,
\end{equation}
\end{proposition}

\begin{proof}

We first observe that the first inequality holds. If $\td'(v) = \td(v)$, simply choose $u = v$ in the left hand side minimization. Else, we have $\td'(v) = \td_i(v)$ for some $i\in \{1,2\}$. Note that the definition of $H_i$ for any $i \in \{1,2\}$ implies that for any $v \in V(G)$ we in fact have 
\begin{align}
\label{eq:einaudi}
d_{H_i}(\sigma, v) 
= \min_{u \in V(G)} \left( \{\td(u)\}_i + d_{H_i[V(G)]}(u, v)  \right)
\end{align}
As the computed distance $\hd_i$ only overestimates $d_{H_i}(\sigma, v)$, we thus have
\begin{align}
    \td_i(v) 
    &= \lfloor \td(v) \rfloor_i + \hd_i(v)
    \ge \lfloor \td(v) \rfloor_i + \min_{u \in V(G)} \left( \{\td(u)\}_i + d_{H_i[V(G)]}(u, v) \right)\\
    &= \min_{u \in V(G)} \left( \td(u) + d_{H_i[V(G)]}(u, v) \right)
\end{align}
where the last equality follows from the fact that there is a path from $u$ to $v$ in $H_i$ only if $\lfloor \td(u) \rfloor_i = \lfloor \td(v) \rfloor_i$. The first inequality of \eqref{eq:smoothformula} then follows from the fact that the distance in $H_i$ are scaled by $\alpha(1+\eps)$. 


So it remains to show the second inequality.
Let $u$ be the minimizer of the right-hand side of \eqref{eq:smoothformula}, and let $P$ be any shortest path from $u$ to $v$. Since we minimize over all nodes of $V(G)$ including $v$, we conclude that
$$ \tilde{d}(u) + \alpha(1+\eps) d_G(u, v) \leq \tilde{d}(v) .$$
This means that $P$ is a smoothing path.

By \Cref{cl:contained} this path is completely contained in $H_i$ for some $i \in \{1,2\}$. 
It follows that for that $i$ we have
\begin{align*}
    \tilde{d}'(v) &\leq \tilde{d}'(v)_i \leq \lfloor \tilde{d}(u)\rfloor_i + (1+\eps/100)\left( \{ \tilde{d}(u)\}_i + \alpha (1+\eps) d_G(u, v)  \right)\\
    & \leq \tilde{d}(u) + \eps\{\tilde{d}(u)\}_i/100 + \alpha(1+\eps)(1+\eps/100)d_G(u, v)
\end{align*}
where we used $\lfloor \td(u) \rfloor_i = \lfloor \td(v) \rfloor_i$. 
Notice that we can bound $ \eps\{\tilde{d}(u)\}_i/100 \le \eps \omega/100 \le \delta/10$ and since $P$ is smoothing, \Cref{cl:short_smoothing_paths} yields $\alpha(1+\eps)\frac{\eps}{100} d_G(u,v) \le \delta/50$. Thus, above inequality simplifies to 
\begin{align*}
    \tilde{d}'(v) &\leq 
    \tilde{d}(u) + \alpha(1+\eps)d_G(u, v) + \delta/2
\end{align*}
and we are done. 
\end{proof}

\begin{proposition}\label{prop:smooth}
If \Cref{alg:smoothing_lemma} is called with an $(\alpha, \delta)$-smooth distance estimate $\tilde{d}$, then the output $\tilde{d}'$ is $(\alpha(1+\eps), \delta/2)$-smooth.
\end{proposition}
\begin{proof}
Let $v, w$ be any vertices in $G$. Letting $u := \argmin_{u' \in V(G)} \tilde{d}(u')+\alpha(1+\eps)d_G(u', v)$ it follows by \Cref{prop:sandwich} that
$$ \tilde{d}(u) + \alpha(1+\eps)d_G(u, v) \leq \tilde{d}'(v),$$
and
$$ \tilde{d}'(w) \leq \tilde{d}(u)+\alpha(1+\eps)d_G(u, w)+\delta/2.$$

Combining these, we get that
\begin{align*}\tilde{d}'(w)-\tilde{d}'(v) &\leq \alpha(1+\eps)d_G(u, w) + \delta/2 - \alpha(1+\eps)d_G(u, v)\\
&\leq \alpha(1+\eps)\left(d_G(u, v)+d_G(v, w)\right) + \delta/2 - \alpha(1+\eps)d_G(u, v) \\ 
&=  \alpha(1+\eps) d_G(v, w) + \delta/2.
\end{align*}
as desired. 
\end{proof}

After \Cref{alg:smoothing_lemma} is analysed, it is simple to finish the proof of the main theorem \Cref{thm:smooth}. 

\begin{proof}
\Cref{prop:smooth} implies by induction that $\td_i(u)$ is $\left( (1+\eps)^i, n\Delta/2^i\right)$-smooth distance estimate. 
Using our choice of $\eps = \eps_0/(10t)$ we get that $(1+\eps)^t \le 1 + \eps_0/2$. 
On the other hand, $n\Delta/2^t < \eps_0/2$. As we assume that all distances in $G$ are at least $1$, we conclude that $\td_t(v) - \td_t(u) \le (1+\eps_0/2)d_G(u,v) + \eps_0/2 \le (1+\eps_0)d_G(u,v)$, as needed. 

We note that we can assume that $\eps_0 = \Omega(1/(n\Delta))$ since otherwise one call to approximate distances already gives exact distances as the answer. Together with the assumptions that edge weights are polynomially bounded, we conclude that $O(\log (n\Delta/\eps_0)) = O(\log n)$. 

It remains to argue about the preservation of the tree-like property and undirected graphs. We will first observe that whenever $\td$ is tree-like in \Cref{alg:smoothing_lemma}, so is $\td'$. To see this, consider any node $v \in V(G)$. 

If $\td'(v) = \td(v)$, we use the tree-likeness of $\td$ to find $u$ with $\td(u) \le \td(v) - \ell_G(u, v)$ and the fact $\td'(u) \le \td(u)$ shows that $u$ also certifies the tree-likeness of $\td'$. 

On the other hand, assume that $\td'(v) = \td_i(v) < \td(v)$ for some $i \in \{1,2\}$. Use the tree-likeness of $\hd_i$ to find some $u$ with $\hd_i(u) \le \hd_i(v) - \ell_{H_i}(u, v)$. 
Note that if, $u = \sigma$, the tree-likeness of $\hd_i$ implies
\begin{align}¨
\label{eq:horrible}
0 = \hd_i(\sigma) \le \hd_i(v) - \ell_{H_i}(\sigma, v) \le \td_i(v) - \td(v)
\end{align}
where for the second inequality we used $\td_i(v) = \lfloor \td(v) \rfloor_i + \hd_i(v)$ and $\td(v)  \le \lfloor \td(v) \rfloor_i + \ell_{H_i}(\sigma, v) $. 
We thus get $\td_i(v) \ge \td(v)$ which contradicts our assumption of $\td_i(v) < \td(v)$. So $u$ has to be some node of $V(G)$ with $\lfloor \td(u) \rfloor_i = \lfloor \td(v) \rfloor_i$. But this implies $\td_i(u) \le \td_i(v) - \ell_G(u, v)$ and as $\td'(u) \le \td_i(u)$, the tree-likeness for $v$ follows. 

Finally, one can check that all above results work also for undirected graphs. There, in the definition of $H_i$ we simply add an undirected edge instead of a directed one. The undirectedness of the added edges may decrease distances between nodes of $H_i$. However, notice that in the proof above, whenever we argue about $H_i$, we only argue about $d_{H_i[V(G)]}(u,v)$ or about $d_{H_i}(\sigma, u)$ for some $u,v \in V(G)$. These distances are not distorted, hence the same proof applies to the undirected case. 
\end{proof}

\section{Tree-Constructing Algorithm}
\label{sec:trees} 

This section is dedicated to prove \Cref{thm:tree}.

\begin{restatable}{theorem}{tree}{Obtaining tree-likeness}
  \label{thm:tree}
  Given an undirected graph $G$ with (real) weights in $[1,\poly(n)]$, a source $s \in V(G)$, accuracy $\eps \in (0,1]$, $D = \poly(n)$ an upper bound on the diameter of $G$, and an approximate oracle $\fO$, \Cref{alg:tree_constructing} computes a tree $T$ such that the distance estimate $\td$ from $s$ induced by $T$ is $(1+\eps)$-approximate (and tree-like) in $G$, using $O(\log n)$ calls to a $(1 + O(\eps/\log n))$-approximate distance oracle and $O(\log^2 n)$ calls to a $1.1$-approximate distance oracle.
\end{restatable}

We refer the reader to the discussion after \Cref{lem:tree-like_vs_tree} to see that in this setup constructing tree-like distances is really the same as constructing an approximate shortest path tree $T$. We in fact do the latter. 

The algorithm is based on a simple idea: if $G$ has radius $D$ from $s$, we ``split'' the graph in the middle into the set $S$ of nodes with distance roughly $D/2$ from $s$ and those that are further. 
We get a graph with radius roughly $D/2$ where we solve the problem recursively and after this is done we ``stitch'' the trees together to get the final output tree (see \Cref{fig:tree-algorithm}). 

There is an important issue with this approach. Constructing the tree of $S$ and of $V(G) \setminus S$ can work only if there is no shortest path that would repeatedly go outside and inside of $S$. If this case happens, stitching the trees together can seriously distort distances. 

Therefore, the set $S$ cannot be simply chosen by first calling an approximate distance oracle and then adding all nodes with a distance estimate of at most $
D/2$ to $S$. 
Instead, we need to use as $S$ the solution of the following ball cutting problem. 

\begin{definition}[Simple ball cutting problem]
    \label{def:cutting}
    Given an undirected graph $G$ with nonnegative edge weights, a non-empty source set $S \subseteq V(G)$, and parameters $D > 0$ and $\eps > 0$. 
    We ask for a set $S^{sup} \subseteq V(G)$ with $S^{sup} \supseteq S$ such that
    \begin{enumerate}
        \item for every $v \in V(G)$, whenever $d_G(S, v) \le D$, we have $v \in S^{sup}$ and
        \item for every $v \in S^{sup}$ we have $d_{G[S^{sup}]}(S, v) \le (1+\eps)\min(d_G(S, v),D)$.
    \end{enumerate}
\end{definition}

Note that this problem would be easy to solve with a tree-like approximate distance oracle: then the simple approach of running the approximate distance oracle once and taking all nodes with approximate distance at most $(1+\eps)D$ works.

    All algorithms in the rest of the section work also for directed graphs but we state them only for undirected graphs, as for directed graphs we can get exact distances and exact distances are clearly tree-like. 
\subsection{Ball Cutting}
\label{subsec:cutting}

In this subsection we show how to solve the ball cutting problem from \Cref{def:cutting}. 

\begin{lemma}
\label{lem:cutting}
Given an undirected graph $G$ with minimum edge weight at least $1$, a source set $S \subseteq V(G)$, accuracy $\eps \in (0,1]$, $D > 0$, and an approximate oracle $\fO$, \Cref{alg:ball_cutting} solves the simple ball cutting problem from \Cref{def:cutting} (with inputs $G,S,D$ and $\eps$), using a single call to a $(1+\eps/4)$-approximate distance oracle and $O(\max(\log D,1))$ calls to a $1.1$-approximate distance oracle.
\end{lemma}

We use the following algorithm. 

\begin{algorithm}[ht]
	\caption{Ball cutting algorithm $\textsc{BallCutting}(G, S, \eps, D)$}
	\label{alg:ball_cutting}
	{\bf Input:} An undirected graph $G$ with minimum edge weight at least $1$, a source set $S \subseteq V(G)$, parameters $\eps \in (0,1]$ and $D > 0$ \\
        {\bf Oracle:} For $S' \subseteq V(G)$, $\fO(G, S', \eps')$ returns $(1+\eps')$-approximate distance estimates in the graph we obtain from $G$ by adding a supersource node and connect it to each node in $S$ with an edge of length $0$.   \\
	{\bf Output:} A set $S^{sup}$ satisfying properties from \Cref{def:cutting}. \\
	\begin{algorithmic}[1]
            \State $t \leftarrow O(\max(\log D,1))$
            \State Define $\td_0 \leftarrow \fO(G, S, \eps/4)$
            \State $S_0 = \{v \in V(G) : \td_0(v) \le (1+\eps/4)D\}$
            \For{$i \leftarrow 1, \dots, t$}
                \State $\td_i(u) \leftarrow \fO(G, S_{i-1}, 0.1)$
                \State $S_i = \{v \in V(G) : \td_i(v) \le 0.5^{i+1} \cdot \eps D \}$
            \EndFor
            \State \Return $S^{sup} \leftarrow S_t$
	\end{algorithmic}
\end{algorithm}

We first observe that clearly whenever $d_G(S, v) \le D$, we have $\td_0(v) \le (1+\eps/4)D$ and hence $v \in S_0$. So the first property in \Cref{def:cutting} is satisfied; we need to focus on the other one. 

We also observe that we have $S_{t-1} = S_{t}$ since to construct $S_t$ we add nodes such that their approximate distance to $S_{t-1}$ is at most $(0.5)^{(t-1)+1} \eps D \ll 1$. 

We note that for any $v \in S_i$ it can even be $d_{G[S_i]}(S, v) = +\infty$. However, we prove that this problem is always fixed in the following step by enlarging $S_i$ to $S_{i+1}$.
This is made precise in the following two claims.

\begin{claim}
\label{cl:cutting0}
For any $v \in S_0$ we have
\begin{align}
    d_{G[S_1]} (S, v) = d_G(S, v)
\end{align}
\end{claim}
\begin{proof}
Let $v \in S_0$. By definition of $S_0$ we have $\td_0(v) \le (1+\eps/4)D$, hence $d_G(S, v) \le \td_0(v) \le (1+\eps/4)D$. 
Let $P$ be a shortest path from $S$ to $v$ of length $d_G(S, v)$. 
Note that $P$ is not necessarily a subset of $S_0$. 
However, since every node $u$ with $d_G(S, u) \le D$ is certainly contained in $S_0$ because $\td_0(u) \le (1+\eps/4) d_G(S,u) \le (1+\eps/4)D$, we have for every node $w \in P$ that $d_G(S_0, w) \le (\eps/4) D$. 
This implies that $\td_1(w) \le 1.1 \cdot (\eps/4) D \le 0.5^{0+1}\eps D$, hence $w \in S_1$. We conclude that $P \subseteq S_1$ and we are done. 
\end{proof}

\begin{claim}
\label{cl:cutting}
Suppose that $v \in S_i$ for any $0 \le i \le t-1$. Then, 
\begin{align}
\label{eq:geometrical}
    d_{G[S_{i+1}]}(S, v) \le D + (1-0.5^{i+1})\eps D.
\end{align}
\end{claim}
\begin{proof}
We prove the claim by induction. For $i = 0$ it follows from \Cref{cl:cutting0} and the fact that any $v \in S_0$ has $d_G(S, v) \le \td_0(v) \le (1+\eps/4)D \leq D + (1-0.5^{0+1})\eps D$.

Next, consider some $i \in \{1,2,\ldots,t-1\}$. 
Let $P$ be any shortest path from some $u \in S_{i-1}$ to $v$ in $G$. 
By definition of $S_i$ and the fact that $\td_i$ is noncontractive, we have $d_G(u,v) = d_G(S_{i-1}, v) \le 0.5^{i+1}\eps D$. 
We now proceed as in the proof of \Cref{cl:cutting0} and prove that $P \subseteq S_{i+1}$. 

To see this, note that any $w \in P$ has $d_G(S_{i-1}, w) \le d_G(S_{i-1}, v) \le \td_i(v) \le  0.5^{i+1}\eps D$. All nodes $w \in P$ with $d_G(S_{i-1}, w) \le \frac{0.5^{i+1}\eps D}{1.1}$ have however $\td_i(w) \le 0.5^{i+1}\eps D$ and hence are taken to $S_i$. 
Since $P$ is a shortest path from $S_i$ to $v$, this implies that any $w \in P$ satisfies $d_G(S_i, w) \le 0.5^{i+1}\eps D - \frac{0.5^{i+1}\eps D}{1.1} \le 0.5^{i+3}\eps D$. This means that $\td_{i+1}(w) \le 1.1d_G(S_i, w) \le 0.5^{i+2}\eps D$ and hence $w \in S_{i+1}$. 

Since $P \subseteq S_{i+1}$, we can now infer that $d_{G[S_{i+1}]}(u, v) = d_G(u, v) \le 0.5^{i+1}\eps D$. 
We finish by using the induction hypothesis on $u$; we have
\begin{align}
    d_{G[S_{i+1}]}(S, v) 
    \le d_{G[S_i]}(S, u) + d_{G[S_{i+1}]}(u, v)
    \le D +  (1 - 0.5^{(i-1) + 1} + 0.5^{i+1}) \eps D = D + (1 - 0.5^{i+1})\eps D
\end{align}
and we are done. 
\end{proof}

\begin{corollary}
The second property from \Cref{def:cutting} is satisfied; that is, for every $v \in S$ we have $d_{G[S^{sup}]}(S, v) \le (1+\eps)\min(d_G(S, v),D)$.
\end{corollary}
\begin{proof}
For any $v \in S^{sup}$ let $i$ be the first index such that $v \in S_i$. Note that $i \le t-1$ since we already observed that $S_{t-1} = S_t$. If $i = 0$, we use \Cref{cl:cutting0} to conclude that $d_{G[S^{sup}]}(S, v) = d_G(S, v)$. Since $\td_0$ is noncontractive we also have $d_{G[S^{sup}]}(S, v) \le (1+\eps/4)d_G(S, v)$. 
Otherwise, if $i \ge 1$, we have $d_G(S, v) > D$ and therefore we can directly use \Cref{cl:cutting}.
\end{proof}

Our procedure for computing tree-like distance estimates actually requires a ball cutting algorithm that works on graphs with slightly more general edge weights. In particular, all of the edge weights are at least one except that edges incident to the source might have a nonnegative weight less than one. \Cref{alg:ball_cutting_more_general} given below uses the ball cutting algorithm given above as a subroutine to handle this slightly more general case.
\begin{lemma}
\label{lem:cutting_general}
Given an undirected graph $G$ with minimum edge weight at least $1$ except that edges incident to source node $s \in V(G)$ can also have edge weights in the interval $[0,1]$, accuracy $\eps \in (0,1]$, $D > 0$, and an approximate oracle $\fO$, \Cref{alg:ball_cutting_more_general} solves the simple ball cutting problem from \Cref{def:cutting} (with inputs $G,\{s\},D$ and $\eps$), using a single call to a $(1+\eps/8)$-approximate distance oracle and $O(\max(\log D,1))$ calls to a $1.1$-approximate distance oracle.
\end{lemma}
\begin{proof}
 First, observe that \Cref{lem:cutting} is true even if we relax the condition that the minimum edge weight is $1$. The only place where we use the minimum weight condition is to prove that $S_{t-1} = S_t$. This condition is even satisfied if the minimum edge weight is at least $\eps$ because $(0.5)^{(t-1) + 1}\eps D \ll \eps$.
 We have to check the two conditions of \Cref{def:cutting}. If $D < 1$, then it follows from the fact that only edges incident to $s$ can have a length of less than one that $S = B_G(s,D)$. Therefore, both conditions of \Cref{def:cutting} are trivially satisfied. Thus, it remains to consider the case that $D \geq 1$.  To check the first condition, consider an arbitrary node $v \in V(G)$. If $d_G(s,v) \leq D$, then also $d_G(S,v) \leq D$ and thus it follows from the fact that $S^{sup}$ is a valid output for the ball cutting problem with input $G,S,\eps/2$ and $D$ that $v \in S^{sup}$. 
 To check the second condition, consider an arbitrary $v \in S^{sup}$. As $S^{sup}$ is a valid output for the ball cutting problem with input $G,S,\eps/2$ and $D$, we get that $d_{G[S^{sup}]}(S,v) \leq (1+\eps/2)\min(d_G(S,v),D)$. We have to show that this implies
 $d_{G[S^{sup}]}(s,v) \leq (1 + \eps)\min(d_G(s,v),D)$.
 First, if $v \in S$, then $d_{G[S^{sup}]}(s,v) = \ell_G(s,v) \leq \min(d_G(s,v),D)$. Next, consider the case that $v \notin S$. If $d_G(s,v) \geq 1$, then 
 \[d_{G[S^{sup}]}(s,v) \leq \eps/2 + d_{G[S^{sup}]}(S,v) \leq \eps/2 + (1+\eps/2)\min(d_G(S,v),D) \leq (1+\eps)\min(d_G(S,v),D).\]
 
 On the other hand, if $v \notin S$ and $d_G(s,v) < 1$, then 
 
 \[d_G(s,v) = d_G(S,v) \leq (1+\eps/2)\min(d_G(S,v),D) \leq (1+\eps)\min(d_G(s,v),D),\]
 
 which finishes the proof.
\end{proof}

\begin{algorithm}[ht]
	\caption{Slightly more general ball cutting algorithm $\textsc{BallCutting'}(G, s, \eps, D)$}
	\label{alg:ball_cutting_more_general}
	{\bf Input:} An undirected graph $G$ with minimum edge weight at least $1$ except that edges incident to source node $s \in V(G)$ can also have edge weights in the interval $[0,1]$, parameters $\eps \in (0,1]$ and $D > 0$ \\
	{\bf Output:} A set $S^{sup}$ satisfying properties from \Cref{def:cutting} with inputs $G,\{s\},D$ and $\eps$. \\
	\begin{algorithmic}[1]
	        \If{$D < 1$} {\Return $S^{sup} \leftarrow \{v \in V(G) \colon \ell_G(s,v) \leq D\}$}
	        \EndIf
            \State $S \leftarrow \{v \in V(G) \colon \ell_G(s,v) \leq \eps/2\}$
            \State \Return $S^{sup} \leftarrow BallCutting(G,S,\eps/2,D)$
	\end{algorithmic}
\end{algorithm}

\subsection{Constructing the Tree}
\label{subsec:tree}

We are now ready to prove \Cref{thm:tree}, the main result of this section.

We construct the approximate shortest path tree using the following recursive algorithm in \Cref{alg:tree_constructing}, see also \Cref{fig:tree-algorithm}.

\begin{figure}
    \centering
    \includegraphics[width = .8\textwidth]{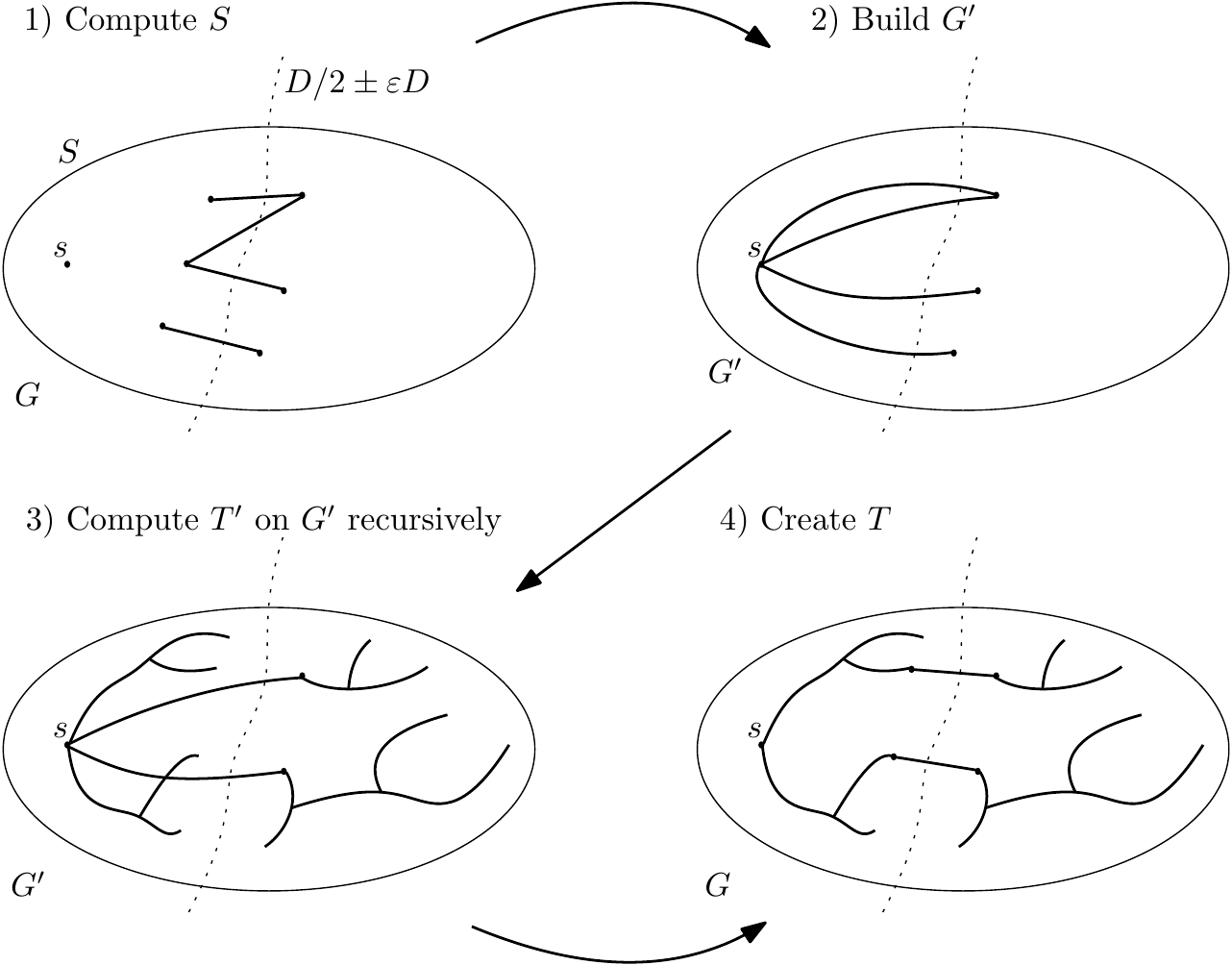}
    \caption{The figure describes \Cref{alg:tree_constructing}. In the first step, we compute the set $S$ that cuts the graph into two halves, both with diameter about $D/2$. The set $S$ has the additional property that any node $v \in V(G)$ can be reached from $s$ by a path that leaves $S$ only at most once such that the path is not much longer than $d_G(s, v)$. \\
    In the second step, we construct a new graph $G'$ where each edge crossing a boundary from some $u \in S$ to some $v \in V(G) \setminus S$ is replaced by a direct edge from $s$ to $v$. The length of the new edge is chosen so that the distance of $v$ to $s$ decreases roughly by $D/2$. \\
    In the third step, we construct the approximate shortest path tree $T'$ recursively in $G'$. \\
    Finally, we replace all edges of $T'$ that are not in $G$ by their counterparts. If all computations were exact, this recursive algorithm would construct a shortest path tree. We lose a $(1+\eps)$-factor in one recursion level due to inexact calculations. In total, we lose $(1+\eps)^{O(\log D)}  = (1+O(\eps \cdot \log D))$ as we choose $\eps < 1/\log D$. }
    \label{fig:tree-algorithm}
\end{figure}

\begin{algorithm}[ht]
	\caption{Tree-constructing algorithm $\textsc{TreeConstructing}(G, D)$}
	\label{alg:tree_constructing}
	{\bf Input:} An undirected graph $G$ with minimum edge weight at least $1$ and diameter upper bound $D$, source $s$, parameter $\eps < O(1/\log D)$, approximate oracle $\fO(G, s, \eps)$. (During recursive calls, edges incident to the source $s$ can have arbitrary nonnegative edge lengths)\\
        {\bf Oracle:} $\fO(G, s, \eps)$ returns $(1+\eps)$-approximate distances $\td : V(G) \to \R_{\ge 0}$ from $s$ in $G$.\\
	{\bf Output:} A tree $T$ such that the distance estimate $\td$ from $s$ induced by $T$ is $(1+ \eps \cdot O(\log D))$-approximate (and tree-like). 
	\begin{algorithmic}[1]
	    \State If $D < 1$ return the star $T$ connecting $s$ with every other node of $G$ \label{line:leaf}
	    \State $\td_0 \leftarrow \fO(G, s, \eps)$
	    \State $S \leftarrow \textsc{BallCutting'}(G, s, \eps, D/2)$
        \State Define $G'$ by replacing each edge $e = \{u,v\}$ with $u \in S, v \not\in S$ by an edge $e_{u,v}$ from $s$ to $v$ with weight $\ell_G(u,v) + \td_0(u) - D/2$\label{line:nonegative}
        \State $T' \leftarrow \textsc{TreeConstructing}(G', (1+\eps)D/2)$
        \State Construct $T$ by replacing each edge $e_{u,v}$ in $T'$ by the edge $\{u,v\}$ \label{line:defT}
        \State \Return $T$
	\end{algorithmic}
\end{algorithm}

We note that the definition of \Cref{line:nonegative} in \Cref{alg:tree_constructing} never creates negative length edges: whenever $\ell_G(u,v) + \td_0(u) - D/2 < 0$, we then have $d_G(s, v) \le d_G(s, u) + \ell_G(u, v) \le \td_0(u) + \ell_G(u,v) < D/2$ where we first used triangle inequality, then the fact that $\td_0$ is noncontractive and then our assumption. Hence, \Cref{def:cutting} implies that $v \in S$. 

During the algorithm we may create edges of length $0$ that connect a node to $s$. 
In fact, we note that whenever $D < 1$, it has to be that $D = 0$ due to all lengths being integers. This means that all nodes are connected to $s$ with a zero-length edge. This justifies \Cref{line:leaf} in \Cref{alg:tree_constructing}: the returned tree $T$ in that case is the (trivial) shortest path tree that connects $s$ with all other nodes by a zero-length edge. 

Let $\td'$ be the distance estimate induced by $T'$ and $\td$ the distance estimate induced by $T$. \Cref{thm:tree} follows by repeated application of the following claim for $O(\log (nD)) = O(\log n)$ steps.

\begin{claim}
\label{cl:tree_main}
Assume that the distance estimate $\td'$ is $(1+\eps_0)$-approximate for $\eps_0 < 1$. Then $\td$ is $(1 + \eps_0 + 5 \eps)$-approximate. 
\end{claim}
\begin{proof}
Consider any $w \in V(G)$. First assume that $w \in S$. In that case by the definition of $\td$ we have $\td(w) = \td'(w)$ and the guarantees of $\td'(w)$ say that $\td'(w) \le (1+\eps_0)d_{G[S]}(s, w)$. 
Using \Cref{def:cutting}, we have $d_{G[S]}(s, w) \le (1+\eps)d_G(s, w)$. We conclude that $\td(w) \le (1+\eps_0)(1+\eps)d_G(s, w)$ as needed. 

Next, we assume $w \not\in S$. Consider a shortest path $P_0$ from $s$ to $w$ and let the edge $\{u_0, v_0\}$ be the last edge of $P_0$ such that $u_0 \in S$ and $v_0 \in V(G) \setminus S$. Similarly, consider the path $P$ that connects $s$ and $w$ in $T$. Let $\{u,v\}$ be the unique edge of it such that $u \in S, v \not\in S$ (see \Cref{fig:tree-like}). Our task is to upper bound the length of $P$ in terms of the length of $P_0$.

Formally, using the definition of $T$ in terms of $T'$ in \Cref{line:defT} we write
\begin{align}
\label{eq:split}
    \td(w) = \td'(u) + \ell_G(u,v) + \td'(w) - \td'(v)
\end{align}
We will bound separately $\td'(w)$ and the rest of the right hand side. See \Cref{fig:tree-like}, the bound of $\td'(w)$ corresponds to the comparison of $P$ and $P_0$ to the right of the boundary between $S$ and $V(G)\setminus S$; the bound on $\td'(u) + \ell_G(u,v) - \td'(v)$ corresponds to the parts of the paths to the left of the boundary. 

\begin{figure}
    \centering
    \includegraphics[width = .6\textwidth]{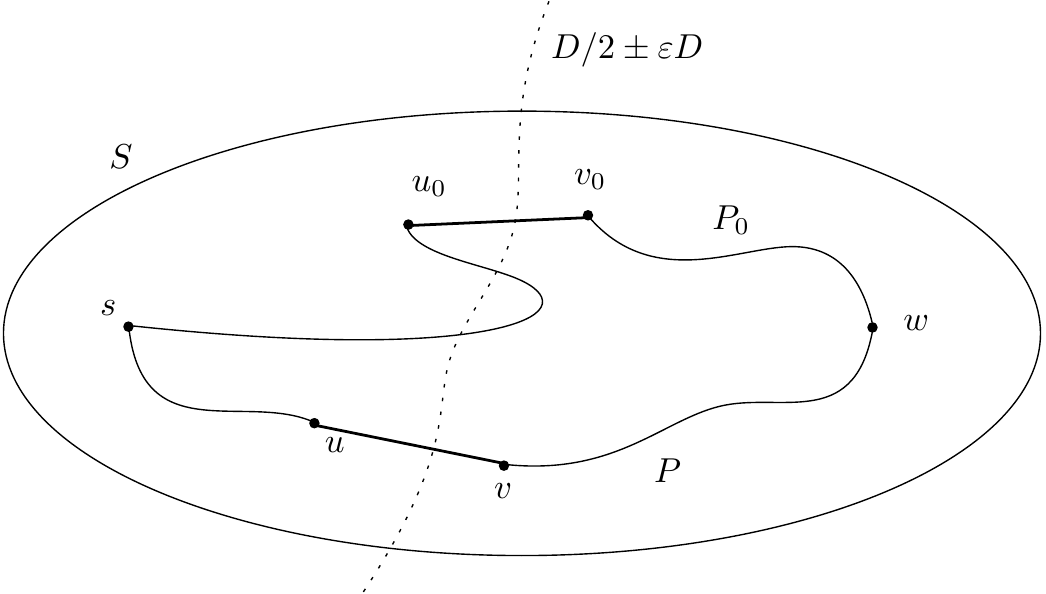}
    \caption{A figure for the proof of \Cref{cl:tree_main}. We consider the optimal path $P_0$ from $s$ to $w$ via $u_0, v_0$ and the path $P$ taken by the constructed tree $T$ that goes through the edge $u,v$.  }
    \label{fig:tree-like}
\end{figure}

{\bf First upper bound: }
To bound $\td'(w)$, we first note that $\td'(w) \le (1+\eps_0)d_{G'}(s, w)$ so we need to bound the value of $d_{G'}(s, w)$. We have
\[
d_{G'}(s, w) 
\le \ell_{G'}(s, v_0) + d_{G}(v_0, w)
= \ell_G(u_0, v_0)  + \td_0(u_0) - D/2 + d_G(v_0, w)
\]
where we used that $\{u_0, v_0\}$ is the last edge of $P_0$ crossing $S$, hence $d_{G[V(G) \setminus S]}(v_0, w) = d_G(v_0, w)$. Moreover, we can bound $\td_0(u_0) \le (1+\eps)d_{G}(s, u_0)$. So, we can write
\[
d_{G'}(s, w) 
\le (1+\eps) \left(
\ell_G(u_0, v_0)  + d_G(s, u_0)  + d_G(v_0, w)\right)- D/2
\]
and hence we can combine above bounds to get
\begin{align}
\td'(w) 
&\le \left( 1 + \eps_0\right)\left(1 + \eps \right) \left( \ell_G(u_0, v_0)  + d_G(s, u_0) + d_G(v_0, w)    \right) - (1+\eps_0)D/2 \\
&=  \left( 1 + \eps_0\right)\left(1 + \eps \right) d_G(s, w) - (1+\eps_0)D/2 \label{eq:first}
\end{align}

{\bf Second upper bound: }
Next, we need to upper bound the rest of the right hand side of \Cref{eq:split}, i.e., the expression $\td'(u) + \ell_G(u,v) - \td'(v)$. To do so, we first write
\[
\td'(u) \le (1+\eps_0)d_{G[S]}(s, u) \le (1+\eps_0)(1+\eps)d_{G}(s, u)
\]
where we used the inductive assumption on $\td'$ together with $d_{G[S]}(s, u) = d_{G'}(s, u)$ and then the property of $S$ from \Cref{def:cutting}. Then we write 
$$\td'(v) = \ell_{G'}(\{s, v\}) 
= \ell_G(u, v) + \td_0(u) - D/2
\ge \ell_G(u, v) + d_G(s, u) - D/2
$$
where we used that $\td_0$ is noncontractive. 
Plugging in, we get
\begin{align}
\td'(u) + \ell_G(u,v) - \td'(v) 
&\le (1+\eps_0)(1+\eps)d_{G}(s, u) + \ell_G(u,v) -  \ell_G(u,v) - d_G(s, u) + D/2\\
&\le (\eps_0 + 2\eps)d_G(s, u) + D/2.\label{eq:second}
\end{align}

{\bf Putting it together}
Putting \Cref{eq:first,eq:second} together, we conclude that 
\begin{align}
\label{eq:datlim}
    \td(w) \le \left( 1 + \eps_0\right)\left(1 + \eps \right) d_G(s, w) - \eps_0 D/2 + (\eps_0 + 2\eps)d_G(s, u) 
\end{align}
Since we have $d_G(s, u) \le (1+\eps)D/2$ by the second property in \Cref{def:cutting}, we can bound the error terms of the right hand side by 
\[
- \eps_0 D/2 + (\eps_0 + 2\eps)(1+\eps)D/2 \le 3\eps \cdot D/2.
\]
Finally, since $d_G(s, w) \ge D/2$ by definition of $S$ (\Cref{def:cutting}), we can simplify \Cref{eq:datlim} to 
\begin{align}
    \td(w) \le \left( 1 + \eps_0 + 3\eps \right) d_G(s, w) .
\end{align}
\end{proof}

\bibliographystyle{alpha}
\bibliography{ref}

\appendices

\section{What Can and Cannot be Done with Strong Distances}
\label{sec:application}

In this section we first show in \Cref{subsec:level_cuts} how the notions of strong distances can be used in practice to greatly simplify proofs of some known results. 
We then, on the other hand, show in \Cref{subsec:mpx} a problem solvable with exact distance oracle that we do not know how to solve even with the power of strong distances on our side. 

\subsection{A Success Story: Ball Growing}
\label{subsec:level_cuts}

Here we give a simple demonstration of the strength of strong distances. 
The problem that we solve is also called ``blurring'' or ``blurry ball growing''. In that problem, we are given a node (or in general a set of nodes) and we want to find a cluster around it of diameter $D$ such that every edge $e$ is cut by this ball only with probability $\ell_G(e)/D$. The cluster should be implementable in a parallel/distributed manner with a few calls to a distance oracle. 

We note that proofs of such results in \cite{becker_emek_lenzen2020blurry_ball_growing} and \cite{elkin_haeupler_rozhon_grunau2022Clusterings_LSST} are nontrivial and somewhat technical.
This is the case even for \cite{elkin_haeupler_rozhon_grunau2022Clusterings_LSST} where the authors work with the approximate distance oracle from \cite{rozhon_grunau_haeupler_zuzic_li2022deterministic_sssp} that, as explained in \Cref{sec:nobody-gets-strong}, gives not only approximate tree-like distances, but also approximate potential as a bonus (so we have both an approximate distance and potentials, but they are two different functions; in contrast a strong distance estimate has both properties at once).

This is in stark contrast with the fact that the result is trivial if we have access to an \emph{exact} distance oracle. 
With our transformation from off-the-shelf approximations to strong distances (\Cref{thm:results_undirected_smooth}) at hand, the proof is now not much more difficult then in the exact distance case as we show next. 

\begin{restatable}{theorem}{blurryInformal}[A variant of Theorem 5 in \cite{becker_emek_lenzen2020blurry_ball_growing}, Theorem 1.5 in \cite{elkin_haeupler_rozhon_grunau2022Clusterings_LSST} ]
\label{thm:blurry_informal}
Given an undirected weighted graph $G$, its node $s$ and a parameter $D > 0$, there is a randomized algorithm computing a set $S$ containing $s$ such that 
\begin{align}
    \label{eq:small}
\max_{v \in S} d_{G[S]}(s,v) \leq D
\end{align}
and moreover every edge $e \in E(G)$ has one endpoint in $S$ and the other in $V(G) \setminus S$ with probability at most
$
O\left(\ell(e)/D\right).
$
The algorithm uses $\poly\log n$ calls to a $(1 + 1/\poly\log n)$-approximate distance oracle. 
\end{restatable}

Let us now show a quick proof of this theorem. 

\begin{proof}
We first observe that the result would be simple to prove should  we have access to an \emph{exact} distance oracle $\fO$. In that case, we simply choose a uniformly random number $\tilde{D} \in [0, D)$ and define $S$ as the set of nodes of distance at most $\tD$ from $s$. This clearly makes \Cref{eq:small} satisfied since whenever a node $v$ is included in $S$, so are all the nodes on the shortest path from $s$ to $v$. Moreover, every edge $e = \{u,v\}$ is cut with probability $|d_G(s,u) - d_G(s,v)|/D \le \ell_G(e)/D$ as needed. 

Going back to the approximate distance oracle, we use \Cref{thm:results_undirected_smooth} with $\eps = 1$ to construct exact distances in some graph $G'$ with weights $\ell_G \le \ell_{G'} \le 2 \ell_G$ and then follow the above algorithm for exact distances but on $G'$ instead of $G$. 
On one hand, we have 
$\max_{v \in S} d_{G[S]}(S,v)  \le \max_{v \in S} d_{G'[S]}(S,v) \leq D$, on the other hand the probability of an edge being cut is now at most $|d_{G'}(s,u) - d_{G'}(s,v)|/D \le \ell_{G'}(e)/D \le 2 \ell_G(e)/D$ and we are done.
\end{proof}

\subsection{A Cautionary Tale: MPX}
\label{subsec:mpx}

We continue by describing a problem which is simple to solve with exact distance oracle but is not directly solved with approximate distances, even in view of our transformations.

The problem is simulating the randomized clustering algorithm of \cite{miller2013parallel} (called MPX next). 
In their algorithm, we have an undirected graph $G$ and a parameter $D > 0$. During the algorithm, each node $v$ first samples a random ``head start'' $h(v)$ which is a number from the exponential distribution with mean $D$. Next, we add a virtual node $\sigma$ and connect it with each node $v \in V(G)$ with an edge of length $O(D\log n) - h(v)$, the first term being there only to make the weight positive. Then we run the exact distance oracle from $\sigma$. The returned shortest path tree defines clusters of $G$ and this is the output of the MPX algorithm. 

It can be proven that, on one hand, with high probability all head starts satisfy $h(v) = O(D \log n)$, hence all edges have positive weight and the diameter of all clusters is at most $O(D \log n)$. 
On the other hand, each edge $e$ is cut by the clustering (that is, its endpoints are in different clusters) with probability $O(\ell(e) / D)$ (compare with \Cref{thm:blurry_informal}). 

It is not clear at all how to simulate this algorithm with the access to even a strongly $(1+\eps)$-approximate-distance oracle, that is, the distance oracle that returns exact distances on a slightly perturbed graph. 

To see this, let us leave the algorithm exactly the same except of using the strong-distance oracle instead of the exact distance oracle. 
The problem is that the distances in the graph $G \cup \{\sigma\}$ are (adversarially) perturbed by only \emph{after} the head starts $h(\cdot)$ are sampled. 
This is not a problem for the cluster diameter, which grows only by a $(1+\eps)$ multiplicative factor.  
However, the guarantee that every edge $e$ is cut only with probability $O(\ell(e)/D)$ does not hold anymore! 

To see this, consider a simple path-graph $G$ consisting of four nodes $u_0, u, v_0, v$ with $\ell_G(u_0, u) = \ell_G(v_0, v) = D$ and $\ell_G(u,v) = 1$. 
With constant probability, the head starts of both $u_0$ and $v_0$ are more than $D$ larger than both head starts $u,v$. 
Moreover, with probability $\Omega(\eps)$ we additionaly have $|h(u_0)- h(v_0)| \le \eps D$. 
In this case the algorithm constructs two clusters rooted in $u_0$ and $v_0$. 
Unfortunately, the adversarial perturbation of edge lengths can be such that in the perturbed graph we always cut the edge $\{u,v\}$. 
This means that we cut $\{u,v\}$ with probability $\Omega(\eps)$, while we require it to be cut only with probability $O(1/D)$ which would be achieved by exact distance oracle.

\section{Undirected APSP Variant of Smoothing}
\label{sec:APSP}

In this section we show how to adapt the algorithm of \Cref{sec:smooth} so that it works in undirected graphs and in the more general context of APSP. We formally prove the following \Cref{thm:smoothAPSP}. 

\begin{restatable}{theorem}{smoothAPSP}
\label{thm:smoothAPSP}
\Cref{alg:smoothing} and \Cref{alg:smoothing_lemma_apsp} compute $(1+\eps)$-approximate smooth APSP distances $\td(u,v)$ in an undirected graph $G$ using $O(\log n)$ calls to a $(1+O(\eps/\log n))$-approximate APSP distance oracle $\fO$ on undirected graphs. 
Moroever,
\begin{enumerate}
    \item If $\fO$ returns tree-like distances, $\td(s, \cdot)$ is also tree-like for every $s \in V(G)$. 
    \item If $\fO$ computes only single source approximate distances, the algorithm solves the single source variant of the problem. 
\end{enumerate}
\end{restatable}  

Notice that the second additional part of our theorem (undirected SSSP) is also a special case of our main theorem \Cref{thm:smooth}. That is, there are two (slightly) different smoothing algorithms for single source undirected smoothing.

The algorithm we use is again \Cref{alg:smoothing} but we replace its inner loop from \Cref{alg:smoothing_lemma} by its slightly modified version \Cref{alg:smoothing_lemma_apsp}. 
Let us first describe the necessary changes to the algorithm, then we analyse it by closely following the analysis of \Cref{sec:directed_algorithm}. 

The main challenge of the APSP case is the construction of the level graph $H_i$. 
In the new construction we add a new node $\sigma_u$ to $G$ for each node $u \in V(G)$, instead of just one node $\sigma$ as in the SSSP case. 
The problem is with cutting edges going between different levels ---  every node $\sigma_u$ would like the graph to be partitioned into different levels and hence would like to cut a different set of edges! 
Fortunately, in the undirected case we can make the smoothing construction work without cutting edges: We simply do not update nodes that are in the ``upper half'' of each level, since the shortest path to those nodes in the level graph can start at the above level (see \Cref{fig:level_graph_apsp} and \Cref{line} in \Cref{alg:smoothing_lemma_apsp}). 

This approach would not work in the directed case where we can have a very short edge $(u,v)$ from a node $u$ with very large $\lf \td(u) \rf$ but $\{\td(u)\} = 0$ to a node $v$ with smaller $\td(v) \ll \td(u)$. If we do not cut $(u,v)$ and run the distance oracle in the level graph $H_i$, we simply learn that $\td_i(s,v)$ is close to $0$, which is a meaningless information. 
Fortunately, this bad case simply cannot occur in undirected graphs: there, a short edge between $u$ and $v$ implies that $u$ lies in the same level as $v$ (unless $v$ lies at the boundary of its level, in which case we need to argue more carefully).   

There are a few other minor changes in \Cref{alg:smoothing_lemma_apsp} with respect to \Cref{alg:smoothing_lemma}. Due to technical reasons, we need three level graphs instead of two and we need to ``slow down'' also the newly added edges connecting $V(G)$ with $V(H_i) \setminus V(G)$, see \Cref{line:defH} in \Cref{alg:smoothing_lemma_apsp}.  

In \Cref{thm:smoothAPSP} we consider the case when the distance oracle returns only SSSP, not APSP distances. This is a special case of the following analysis with APSP oracles, so we next analyse only the case when $\fO$ is an APSP oracle. 

\paragraph{Definitions and the algorithm.}

We define rounding analogously to \Cref{def:rounding}. 

\begin{definition}[Rounding distances]
\label{def:rounding_apsp}
Let $G$ be a graph and $\td$ a distance estimate on it. Let $i \in \{1, 2, 3\}$ and $\omega$ a parameter. 
We define $\{\td(s, v)\}_i$ to be the smallest value such that $\td(s, v) - \{\td(s, v)\}_i - (i-1)\cdot \frac{\omega}{3}$ is divisible by $\omega$. 
We also define $\td(s, v) = \lfloor \td(s, v)\rfloor_i + \{\td(s, v)\}_i$. 
\end{definition}

Next, let us define the level graph analogously to \Cref{def:level_graph} with the main difference being that $H_i$ contains all edges of $G$, i.e., we are not cutting edges (see also \Cref{fig:level_graph_apsp}). 

\begin{figure}
    \centering
    \includegraphics[width = .6\textwidth]{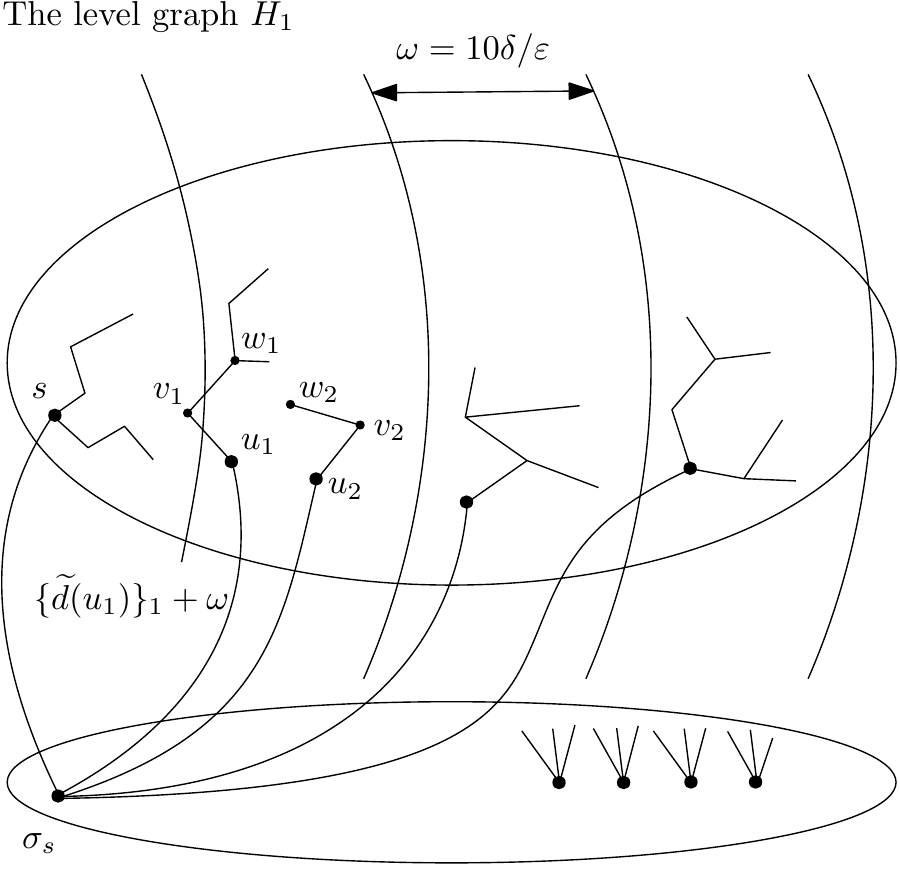}
    \caption{The figure shows the construction of the level graph $H_1$ in the undirected APSP case. We add a node $\sigma_s$ for each $s \in V(G)$. No edges of the original graph $G$ are cut. The edges from $\sigma_s$ to other nodes in $V(G)$ are defined as in the directed case but with additional term $+\omega$ so that $\sigma_s$ is not affected by other nodes $\sigma_{s'}$. \\
    We want to update in \Cref{alg:smoothing_lemma_apsp} only the distance estimate of nodes that lie in the ``first half'' of each level, i.e. $0 \le \{\td(v)\}_1 \le \omega/2$. This is because nodes close to the upper end of a level are affected by short distances to the low end of the next level. To make the new distances tree-like, whenever we update nodes $w_1, w_2$ with $0 \le \{\td(w_1)\}_1, \{\td(w_2)\}_1 \le \omega/2$, we also need to update their predecessors in the tree constructed from the tree-like property of $\hd$. Namely, we update $v_1$ although it lies in a smaller level but we behave as if its level is the same as the level of $u_1$. We also update the value of $v_2$ although $\{\td(v_2)\}_1 > \omega/2$. }
    \label{fig:level_graph_apsp}
\end{figure}\todo{fix picture}

\begin{definition}[Level graph $H_i$ for undirected APSP]
	\label{def:level_graph_apsp}
    Given an input undirected graph $G$, a width parameter $\omega$ and $i \in \{1, 2, 3\}$, we define the \emph{APSP level graph} $H_i$ as follows. 

We define $V(H_i) = V(G) \cup \bigcup_{u \in V(G)}\sigma_u$. 
    The set of edges $E(H_i)$ is a superset of $E(G)$, with $\ell_{H_i} (u,v) = \alpha (1+\eps) \ell_G(u,v)$.     
    Additionally, for each $\sigma_u$ and $v$ we add an edge $\{\sigma_u, v\}$ to $E(H_i)$ for any $u \in V(G)$ of weight $$(1+\eps/10^3)\{\td(u, v)\}_i + \omega.$$ 
\end{definition}

Comparing with \Cref{def:level_graph}, we have additional term $+\omega$ in the definition of edges $\{\sigma_u, v\}$. This is because we do not want the shortest paths from $\sigma_u$ to nodes of $V(G)$ be affected by other edges $\{\sigma_{u'}, v\}$ in $H_i$. This would not be necessary if we were allowed to make the edges of type $\{\sigma_u, v\}$ directed. 
The term $+\omega$ implies additional distortion of distances, but since we choose $(1+\eps/10^4)$-approximate oracle and $\eps\omega/10^4 \ll \delta/2$, this additional error term does not create new difficulties. 

We also ``slow down'' the edge by $(1+\eps/10^3)$. This is only needed to argue about the final distances being tree-like at the end of the proof.

\begin{algorithm}[ht]
	\caption{Partial Smoothing algorithm for undirected APSP}
	\label{alg:smoothing_lemma_apsp}
	{\bf Input:} An undirected graph $G$, parameter $\eps$, approximate APSP oracle $\fO$, $(\alpha,\delta)$-smooth distance estimate $\td(\cdot, \cdot)$.\\
        {\bf Oracle:} $\fO(G, \eps_0)$ returns $(1+\eps_0)$-approximate APSP distances $\hd : V(G) \times V(G) \to \R_{\ge 0}$ in $G$.\\
	{\bf Output:} $((1+2\eps)\alpha, \delta /2)$-smooth distance estimate $\td'(\cdot, \cdot)$
	\begin{algorithmic}[1]
            \State Define $H_1, H_2, H_3$ as in \Cref{def:level_graph_apsp} for $\omega = 100\delta/\eps$ and $\ell(\sigma_s, u) = (1+\eps/10^3)\{\td(s, v)\}_i + w$ \label{line:defH}
            \For{$i \leftarrow 1, 2, 3$}
                \State $\hd_i(\cdot, \cdot) \leftarrow \fO(H_i, \eps/10^4)$
                \For{all $s \in V(G)$}
                    \For{all $v \in V(G)$} \label{line:inner_start}
                        \If{$0 \le \{\td(s,w)\}_i < \omega/2$} \label{line}
                            \State $\td_i(s, v) \leftarrow \lfloor \td(s, u) \rfloor_i + \hd_i(\sigma_s, v) - \omega$
                        \Else
                            \State $\td_i(s,v)  \leftarrow +\infty$ \label{line:inner_end}
                        \EndIf                    
                    \EndFor
                    \EndFor
            \EndFor
            \State \Return $\td'$ defined by $\td'(s, u) = \min\left( \td(s, u), \td_1(s, u ), \td_2(s, u ), \td_3(s, u) \right)$
	\end{algorithmic}
\end{algorithm}

We will now analyse \Cref{alg:smoothing_lemma_apsp}. 
The analysis is very similar to the directed SSSP case from \Cref{sec:directed_algorithm} but we need to argue a bit more carefully about interactions of different levels of the level graph $H_i$. 

Analogously to \Cref{sec:directed_algorithm}, for parameters $\eps, \alpha$ and a node $s \in G$ we define a path $P \subseteq G$ \emph{smoothing} whenever 
\begin{align}
    \label{def:smoothing_apsp}
    \td(s, v) \ge \td(s, u) + \alpha(1+\eps/2)\ell_G(P)
\end{align}
and use its following properties.

\begin{claim}{Analogue of \Cref{cl:short_smoothing_paths}}
\label{cl:short_smoothing_paths_apsp}
Any smoothing path $P$ satisfies $\ell_G(P) \leq 2\delta/(\alpha\eps).$
\end{claim}
\begin{proof}
The proof is the same as in \Cref{cl:short_smoothing_paths}: We have by definition that a smoothing path $P$ from $u$ to $v$ satisfies
$$\alpha(1+\eps/2)\ell_G(P) \leq \tilde{d}(s, v)-\tilde{d}(s, u). $$
By assumption that $\tilde{d}(s, \cdot)$ is $(\alpha, \delta)$-smooth, we have
$$ \tilde{d}(s, v)-\tilde{d}(s, u) \leq \alpha\cdot d_G(u, v) + \delta \leq \alpha \cdot \ell_G(P) + \delta.$$
Combining these gives the desired inequality.
\end{proof}

\begin{claim}{Analogue of \Cref{cl:contained}}
\label{cl:contained_apsp}
Let $P$ be a smoothing path. Let $w_{\max} = \argmax_{w\in P}\,\tilde{d}(w)$  and $w_{\min} = \argmin_{w\in P}\,\tilde{d}(w)$. 
Then $ \tilde{d}(w_{\max})-\tilde{d}(w_{\min}) \le 6\delta/\eps$. 
\end{claim}
\begin{proof}
We have
\begin{align*}
    \tilde{d}(w_{\max})-\tilde{d}(w_{\min}) 
    &\leq \tilde{d}(w_{\max})-\tilde{d}(w_{\min}) + \tilde{d}(v)-\tilde{d}(u) \\
    & = \left(\tilde{d}(w_{\max})-\tilde{d}(u)\right) + \left(\tilde{d}(v)-\tilde{d}(w_{\min})\right)\\
    & \leq \alpha \ell_G(P) + \delta + \alpha \ell_G(P) + \delta \\
    & \leq 2(2\eps^{-1}+1)\delta \le 6\delta/\eps. 
\end{align*} 
where we first used the $(\alpha, \delta)$-smoothness assumption on $(u, w_{\max})$ and $(w_{\min}, v)$ and then we applied \Cref{cl:short_smoothing_paths_apsp}. 
\end{proof}

The fact that smoothing paths cannot be long then implies that both endpoints of any smoothing path are both in the ``bottom half'' of at least one layer of $H_1, H_2$ or $H_3$. The following claim is the place where we need three graphs instead of two as in \Cref{sec:directed_algorithm}. 

\begin{claim}
\label{cl:contained3}
Let $P$ be any smoothing path from $u$ to $v$. Then, there exists $i \in \{1,2,3\}$ such that 
\[
0 \le \{ \td(s, u) \}_i \le \{ \td(s, v) \}_i < \omega/2.
\]
\end{claim}
\begin{proof}
Choose $i$ to be such that $0 \le \{ \td(s, u) \}_i < \omega/3$. Then, the $(\alpha, \delta)$-smoothness of $\td$ and \Cref{cl:short_smoothing_paths_apsp} yield $\td(s,v) \le \td(s,u) + \alpha \ell_G(P) + \delta < \omega/3 + \alpha ( 2\delta/(\alpha \eps)) + \delta \le \omega/3 + 3\delta/\eps$. We use the definition of $\omega = 100\delta/\eps$ to get that $3\delta/\eps < \omega/6$ and thus $\{ \td(s,v) \}_i < \omega/2$ as needed. 
\end{proof}

The main difference to the proof of \Cref{thm:smooth} is that we need to argue more carefully about shortest paths in $H_i$ behaving as in \Cref{fig:level_graph_apsp}. 

\begin{claim}{Structure of shortest paths in $H_i$}
\label{cl:structure_apsp}
Given $s, v \in V(G)$, let $P \subseteq H_i$ be either
\begin{enumerate}
    \item a shortest path from $\sigma_s$ to some $v \in V(G)$ in $H_i$. 
    \item if $\hd(\sigma_s, \cdot)$ is tree-like, $P$ may also be the path that we get by starting at $v$ and each time going from current node $u$ to a node $w$ with $\hd_i(s, w) \le \hd_i(s, u) - \ell_{H_i}(w, u)$ until we reach $\sigma_s$. 
\end{enumerate}
Let $\{\sigma_s, u'\}$ be the first edge of $P$ and let us write $P = \{\{\sigma_s, u'\}\} \cup P'$. Then 
\begin{enumerate}
    \item $P' \subseteq G$. 
    \item $\lfloor \td(s, u') \rfloor_i = \lfloor \td(s, v)  \rfloor_i$. 
    \item $P'$ is a smoothing path from $\sigma_s$ to $v$. 
\end{enumerate}

\end{claim}
\begin{proof}
Assume that $P$ is the shortest path from $\sigma_s$ to $v$, we discuss the other possibility from the claim statement at the end. We first note that $P' \subseteq G$. Otherwise, the length of $P$ is at least $3\omega$ as it has to contain at least $3$ edges from $V(G) \times (V(H_i) \setminus V(G))$, while the edge from $\sigma_s$ to $v$ has length only at most $2\omega$, contradicting that $P$ is a shortest path. 

Next, we prove that $\lfloor \td(s, u') \rfloor_i = \lfloor \td(s, v)  \rfloor_i$. To that end, first assume that $\lfloor \td(s, u') \rfloor_i < \lfloor \td(s, v)  \rfloor_i$. In this case, we have on one hand that 
\begin{align*}
    \ell_{H_i}(P) \le \ell_{H_i}(\sigma_s, v) = (1+\eps/10^3)\{\td(s, v)\}_i + w
\end{align*}
because there is an edge between $\sigma_s$ and $v$. 
But on the other hand we have
\begin{align*}
\ell_{H_i}(P) 
&= \ell_{H_i}(\sigma_s, u') + \ell_{H_i}(P') \\
&\ge (1+\eps/10^3)\{\td(s, u')\}_i + w + \alpha(1+\eps)\frac{\td(s, v) - \td(s, u') - \delta}{\alpha}
\end{align*}
where we used $(\alpha, \delta)$-smoothness of $u',v$. 
If we put these two bounds together, we get
\begin{align*}
(1+\eps/10^3)\{\td(s, v)\}_i \ge (1+\eps/10^3)\{\td(s, u')\}_i +  (1+\eps)(\td(s, v) - \td(s, u') - \delta)
\end{align*}
which implies that
\begin{align*}
\{\td(s, u')\}_i + (1+\eps) (\lfloor \td(s, v) \rfloor_i - \td(s, u')) - 2\delta \le 0
\end{align*}
and we can simplify it further to
\begin{align*}
    \eps ( \lfloor \td(s, v) \rfloor_i - \lfloor \td(s, u') \rfloor_i ) \le 2\delta
\end{align*}
Finally, our assumption on $u'$ then implies $\eps \omega \le 2\delta$, a contradiction with the definition of $\omega$. 

Next, assume that $\lfloor \td(s, u') \rfloor_i > \lfloor \td(s, v) \rfloor_i$. We again derive a contradiction. 
As a small aside, we note that the following argument is the place where we rely on \Cref{line:inner_end} in \Cref{alg:smoothing_lemma_apsp}, i.e., we set $\td_i(s, v) \le +\infty$ only for $0 \le \{\td(s, v) \}_i < \omega/2$. Also, this is a place where we crucially use the fact $G$ is undirected. 

As in the previous case we start by writing 
\begin{align*}
    \ell_{H_i}(P) \le \ell_{H_i}(\sigma_s, v) = (1+\eps/10^3)\{\td(s, v)\}_i + w
\end{align*}
and similarly to the previous case we write
\begin{align*}
\ell_{H_i}(P) 
&= \ell_{H_i}(\sigma_s, u') + \ell_{H_i}(P') \\
&\ge (1+\eps/10^3)\{\td(s, u')\}_i + w + \alpha(1+\eps)\frac{\td(s, u') - \td(s, v) - \delta}{\alpha}.
\end{align*}
Notice that for the second bound we use the fact that as $G$ is undirected, $P'$ is also the shortest path from $v$ to $u'$. 
We again put the two bounds together and simplify, we get
\begin{align*}
    (1+\eps/10^3)\{\td(s, v)\}_i
    \ge (1+\eps)(\td(s, u') - \td(s, v) - \delta)
\end{align*}
Note that $\{ \td(s, v) \}_i < \omega/2$ together with our assumption on $u'$ implies $\td(s, u') - \td(s, v) \ge \omega/2$, hence we can simplify to 
\begin{align*}
    (1+\eps/10^3)\omega/2
    \ge (1+\eps)(\omega/2 - \delta)
\end{align*}
which is, however, in contradiction with our definition $\omega = 100\delta/\eps$. This finishes the proof that $\lfloor \td(s, u') \rfloor_i = \lfloor \td(s, v)  \rfloor_i$ (a second bullet point in the claim statement). 

We finish by showing that $P'$ is a smoothing path. We again start by writing
\begin{align*}
\ell_{H_i}(P) \le (1+\eps/10^3)\{\td(s, v)\}_i + w
\end{align*}
which in turn implies
\begin{align*}
\ell_{H_i}(P') \le (1+\eps/10^3)\{\td(s, v)\}_i - (1+\eps/10^3)\{\td(s, u')\}
\end{align*}
Using the fact that lengths in $H_i$ are multiplied by the factor of $\alpha(1+\eps)$, we conclude that 
\begin{align*}
    \{\td(s, v)\}_i \ge \{\td(s, u')\}_i + \alpha \frac{1+\eps}{1+\eps/10^3} \ell_{H_i}(P' )
\end{align*}
which implies the smoothness of $P' $ as needed. 

Finally, consider $P$ to be not the shortest path but the path defined in the second bullet point of the claim statement. The proof for this case is very similar and we thus omit it. We note that  in calculations bounding $\ell_{H_i}(P)$, we need to replace the bound $$\ell_{H_i}(P) \le (1 + \eps/10^3) \{\td(s, v)\}_i + w$$ by  $$\ell_{H_i}(P) \le (1 + \eps/10^4) \left( (1 + \eps/10^3) \{\td(s, v)\}_i + w \right).$$
This leads to a small additional error term of size at most $\frac{\eps}{10^4} \cdot 3w$ that does not affect the conclusion of above argument. 
\end{proof}

We now follow the single source case and prove an analogue of \Cref{prop:sandwich}. 

\begin{proposition}\label{prop:sandwich_apsp} 
    The output function $\td'(s, \cdot)$  in \Cref{alg:smoothing_lemma_apsp} satisfies for any $s$ that 
    \begin{equation}\label{eq:smoothformula_apsp}  
    \min_{u \in V(G)} \left( \tilde{d}(s, u) + \alpha(1+\eps) d_G(u, v) \right) 
    \leq \tilde{d}'(s, v) 
    \leq  \min_{u \in V(G)} \left( \tilde{d}(s, u) + \alpha(1+\eps)  d_G(u, v) \right) + \delta/2,
    \end{equation}
    \end{proposition}
    
    \begin{proof}
    Recall we define $\td'(s, v) = \min \left( \td(s, v), \td_1(s, v), \td_2(s, v), \td_3(s, v) \right)$. 
    We start with the first inequality, i.e., we need to prove that the left hand side of \Cref{eq:smoothformula_apsp}, that we denote as LHS, is upper bounded by $\td(s, v), \td_1(s, v), \td_2(s, v), \td_3(s, v) $. We certainly have $\text{LHS} \le \td(s, v)$ as $v$ is one of the nodes we minimize over in LHS. 
    Hence we need to prove that it is also upper bounded by $\td_i(s, v)$ for each $i \in \{1,2,3\}$. We recall that whenever $\td_i(s,v)$ is not $+\infty$ (in which case there is nothing to prove), it is defined as: 
    \begin{align}
        \label{eq:ap1}
        \td_i(s, v) 
        &= \lfloor \td(s,v)\rfloor_i + \hd_i(\sigma_s, v) - \omega\\
        &\ge \lfloor \td(s,v)\rfloor_i + d_{H_i}(\sigma_s, v) - \omega
    \end{align}
    Recall that by \Cref{cl:structure_apsp} any shortest path $P$ from $\sigma_s$ to $v$ consists of an edge from $\sigma_s$ to some $u'$ with $\lfloor \td(s, u')\rfloor_i = \lfloor \td(s, v) \rfloor_i$ and then a path $P' = P \setminus \{\{\sigma_s, u'\}\} \subseteq G$. This allows us to rewrite the above inequality as
    \begin{align}
    \td_i(s, v) 
    &\ge \lfloor \td(s, u') \rfloor_i + ((1+\eps/10^3)\{\td(s, u')\}_i + \omega) + d_{H_i}(u', v) - \omega\\
    &\ge \td(s, u') +  \alpha (1+\eps)d_G(u', v)
    \ge \text{LHS}
    \end{align}
    as needed. 

    We continue with the second inequality in \Cref{eq:smoothformula_apsp} whose right hand side we denote as RHS. Let $u'$ be the minimizer of RHS. The node $u'$ has to satisfy that $\td(s, u') + \alpha (1+\eps) d_G(u',v) \le \td(s, v)$ as we can plug in $v$ to the minimization in RHS. Hence, the shortest path from $u'$ to $v$ in $G$ is smoothing according to \Cref{def:smoothing_apsp}. Applying \Cref{cl:contained3}, we get existence of $i$ such that $0 \le \{\td(s, u)\}_i \le \{\td(s, v)\}_i < \omega/2$. For this $i$, we have $\td(s,v)_i < +\infty$ and thus we can write
    \begin{align}
        \label{eq:tricky}
        \td'(s, v) 
        \le \tilde{d}(s, v)_i 
        = \lfloor \td(s, v)\rfloor_i + \hd_i(\sigma_s, v) - \omega
    \end{align}
    We continue by using the  fact that our distance oracle is $(1+\eps/10^4)$-approximate to bound $\hd_i(\sigma_s, v) \le (1+\eps/10^4) d_{H_i}(\sigma_s, v)$. Plugging this to above inequality and using \Cref{cl:structure_apsp} that says that the shortest path $P$ from $\sigma_s$ to $v$ can be written as $\{\{\sigma_s, u\}\} \cup P'$ for $P' \subseteq G$, we get
    \begin{align}
    \label{eq:tricky2}
    \td'(s, v)
    \le \lfloor \td(s, v)\rfloor_i  + \left( 1 + \eps/10^4 \right) \cdot 
        \left(
            \min_{u \in V(G)} 
             ( (1+\eps/10^3)\{\tilde{d}(s, u)\}_i + \omega) + \alpha (1+\eps) d_G(u, v)  
        \right)- \omega
    \end{align}
    We need to simplify the right-hand side of this inequality. First, we get rid of the inner term $(1+\eps/10^3)$ by simplifying the inequality as
    \begin{align}
    \label{eq:tricky22}
    \td'(s, v)
    \le \lfloor \td(s, v)\rfloor_i  + \left( 1 + \eps/10^4 \right) \cdot 
        \left(
            \min_{u \in V(G)} 
             ( \{\tilde{d}(s, u)\}_i + \omega) + \alpha (1+\eps) d_G(u, v)  
        \right)- \omega + \frac{2 \eps}{1000}\omega 
    \end{align}
    Next, let $u'$ be a minimizer of the expression on the right-hand side. 
    \Cref{cl:structure_apsp} implies that $        \lfloor \td(s, u')\rfloor_i = \lfloor \td(s, v)\rfloor_i$. Hence, we can rewrite the right hand side of \Cref{eq:tricky2} by isolating the term proportional to $\eps/10^4$, to get
    \begin{align}
        \label{eq:final}
        &\text{RHS of \Cref{eq:tricky2}} \\
        &= 
            \min_{u \in V(G)} \left(
             \tilde{d}(s, u)  + \alpha (1+\eps)d_G(u, v)  
             \right)
             + \frac{\eps}{10^4} \cdot \min_{u \in V(G)} \left( 
         (\{\tilde{d}(s, u)\}_i  + \omega) + \alpha (1+\eps)d_G(u, v)         \right)+ \frac{2 \eps}{1000}\omega 
    \end{align}
    Finally, we bound the term proporitional to $\eps/10^4$ by using that 
    \begin{align}
    \label{eq:final2}
    \min_{u \in V(G)} \left(
    (\{\tilde{d}(s, u)\}_i  + \omega) + \alpha (1+\eps)d_G(u, v) \right) \le \{\td(s, v)\}_i + \omega \le 2\omega 
    \end{align}
    where we used the fact that the minimization includes $v$. 
    Putting \Cref{eq:tricky2,eq:final,eq:final2} together, we conclude that
    \begin{align}
    \label{eq:finalfinal}
    \td'(s, v) \le  \min_{u \in V(G)} \left(
             \tilde{d}(s, u)  + \alpha (1+\eps)d_G(u, v)  
             \right) + \frac{\eps}{10^4} \cdot 2\omega+ \frac{2 \eps}{1000}\omega . 
    \end{align}
    Finally, the fact $\frac{\eps}{10^4} \cdot 2\omega+ \frac{2 \eps}{1000}\omega  < \delta/2$ concludes the proof.

    \end{proof}

\begin{proposition}\label{prop:smooth_apsp} 
    If \Cref{alg:smoothing_lemma_apsp} is called with an $(\alpha, \delta)$-smooth distance estimate $\tilde{d}$, then the output $\tilde{d}'$ is $(\alpha(1+\eps), \delta/2)$-smooth.
\end{proposition}
\begin{proof}

The proof for each $s$ is very similar to the proof of \Cref{prop:smooth}, we rewrite it here for completeness. 

Let $v, w$ be any vertices in $G$. Letting $u := \argmin_{u' \in V(G)} \tilde{d}(u')+\alpha(1+\eps)d_G(u', v)$ it follows by \Cref{prop:sandwich_apsp} that
$$ \tilde{d}(u) + \alpha(1+\eps)d_G(u, v) \leq \tilde{d}'(v),$$
and
$$ \tilde{d}'(w) \leq \tilde{d}(u)+\alpha(1+\eps)d_G(u, w)+\delta/2.$$

Combining these, we get that
\begin{align*}\tilde{d}'(w)-\tilde{d}'(v) &\leq \alpha(1+\eps)d_G(u, w) + \delta/2 - \alpha(1+\eps)d_G(u, v)\\
&\leq \alpha(1+\eps)\left(d_G(u, v)+d_G(v, w)\right) + \delta/2 - \alpha(1+\eps)d_G(u, v) \\ 
&=  \alpha(1+\eps) d_G(v, w) + \delta/2.
\end{align*}
as desired.

\end{proof}

\begin{proposition}
\label{prop:treelike}
The function $\td'(s, \cdot)$ is tree-like for any $s$, provided that the approximate shortest path oracle $\fO$ and the distance function $\td$ are tree-like. 
\end{proposition}
\begin{proof}
Fix $s$. Let $v \not=s$ be an arbitrary node; recall that we need to prove existence of $w_0$ with 
\begin{align}
\label{eq:tree}
\td'(s, w_0) \le \td'(s, v) - \ell_G(w_0, v)    
\end{align}
Also, recall that $\td'(s, v) = \min(\td(s, v), \td_1(s, v), \td_2(s, v), \td_3(s, v))$; if $\td'(s, v) = \td(s, v)$, \Cref{eq:tree} is implied by the assumption that $\td(s, \cdot)$ is tree-like and the fact that $\td'(s, w) \le \td(s, w)$. Next, assume $\td'(s, v) < \td(s, v)$ and let $i$ be such that $\td'(s, v) = \td_i(s, v)$; in particular we have $\td_i(s,v) < +\infty$ in this case. The tree-likeness of our oracle implies that there exists $w$ such that
\begin{align}
\label{eq:tree_have}
\hd_i(s, w) \le \hd_i(s, v) - \ell_{H_i}(v, w). 
\end{align}

We will next go through all possible cases of where the node $w$ possibly lies and argue that we can always set $w_0 = w$ in \Cref{eq:tree}.  

Let us define a path $P$ as in \Cref{cl:structure_apsp}: we start with the node $v$ and each time go from the current node $u$ to its neighbor $u'$ with $\hd_i(s, u') \le \hd_i(s, u) - \ell_{H_i}(u', u)$ until we reach $\sigma_s$. We also write $P = \{\{\sigma_s, u'\}\} \cup P'$ for $P' \subseteq G$. 

Let us start with the case when $P' = \{\}$, i.e., $u' = w$. 
In that case, we use \Cref{eq:tree_have} to observe that $\hd_i(\sigma_s, v) \ge \ell_{H_i}(\sigma_s, v) =  (1+\eps/10^3)\{\td(s, v)\}_i + \omega$. When we plug this in the definition $\td_i(s, v)$, we get
\begin{align}
\td_i(s,v) 
&= \lf \td(s, v) \rf_i + \hd_i(\sigma_s, v) - \omega   \\
&= \td(s,v) + \eps\{\td(s,v)\}_i/100\\
& \ge \td(s,v),
\end{align}
a contradiction with our assumption that $\td_i(s, v) = \td'(s,v) < \td(s,v)$.

Hence, in the rest of the proof we consider the case when $w \in P' \subseteq V(G)$. Here we need to consider several subcases. 
First, assume that $\lfloor \td(s, w) \rfloor_i < \lfloor \td(s, v) \rfloor_i$. In this case we use \Cref{eq:tree_have} to conclude that 
\begin{align}
\hd_i(s, v) 
&\ge \ell_{H_i}(v,w) + \hd_i(s, w)\\
&\ge \ell_{H_i}(v, w) + d_{H_i}(\sigma_s, w) \\
&\ge \ell_{H_i}(v, w) + \omega    
\end{align}
where we used the tree-likeness of $\hd$, the fact $\hd$ is proper, and that any edge from $\sigma_s$ to $V(G)$ has length at least $\omega$. Hence 
\begin{align}
\td_i(s, v)
&=  \lf \td(s, v) \rf_i + \hd_i(\sigma_s, v) - \omega   \\
&\ge \lf \td(s, v) \rf_i + \ell_{H_i}(v, w)  \\
\end{align}
and as our assumption on $w$ implies $ \td(s, w) < \lfloor \td(s, v) \rfloor_i$, we conclude that
\begin{align}
    \td_i(s, v) - \td_i(s, w) 
    &\ge \lf \td(s, v) \rf_i + \ell_{H_i}(v, w)  - \td_i(s, w)
    \ge \ell_{H_i}(v, w),
\end{align}
i.e., we can choose $w_0 = w$ in \Cref{eq:tree} as desired. 

Next, consider the case when $\lfloor \td(s, w) \rfloor_i > \lfloor \td(s, v) \rfloor_i$. This case is in direct contradiction with \Cref{cl:contained_apsp}: If we define $w_{\max} = \argmax_{w\in P}\,\tilde{d}(w)$  and $w_{\min} = \argmin_{w\in P}\,\tilde{d}(w)$, we would have $\td(s, w_{\max}) \ge \td(s, w) \ge \lfloor \td(s, w)\rfloor_i$ and $\td(s, w_{\min}) \le \td(s, v) \le \lfloor \td(s, v) \rfloor_i + \omega/2$ which implies $\td(s, w_{\max}) - \td(s, w_{\min}) \ge \omega/2$ contradicting the bound from \Cref{cl:contained_apsp}. 

So, we have $\lfloor \td(s, w) \rfloor_i = \lfloor \td(s, v) \rfloor_i$. 
As a final case distinction, we consider the cases when $0 \le \{\td(s, w)\}_i < \omega/2$ (and hence $\td(s,w)_i < +\infty$) and when $\omega/2 \le \{\td(s, w)\}_i < \omega$ (and hence $\td(s,w)_i = +\infty$). 

In the first case the tree-likeness follows directly from definition:
\begin{align*}
\td_i(s, v) - \td_i(s,w)
&= \left( \lfloor \td(s,v) \rfloor_i + \hd(s, v) - \omega \right) - \left(  \lfloor \td(s,w) \rfloor_i + \hd(s, w) - \omega\right)\\
&= \hd(s, v) - \hd(s, w)\\
&\ge \ell_{H_i}(v, w) \ge \ell_G(v, w)
\end{align*}

It remains to consider the case when $\omega/2 \le \{\td(s, w)\}_i < \omega$. We will show that this case cannot happen since for $j = (i+1) \mod 3$ we have $\td_j(s, v) < \td_{i}(s, v)$ which contradicts our choice of $i$ as the index for which $\td'(s, v) = \td_i(s, v)$. 
The intuitive reason for this is that the ``slow down'' $(1+\eps/10^3)$ of edges from $V(H_i) \setminus V(G)$ to $V(G)$ affects more the graph $H_i$ than $H_j$. 
For the formal proof we first observe that $\{\td(s, w)\}_i \ge \omega/2$ implies via \Cref{cl:contained_apsp} that $\{\td(s, v)\}_i \ge \omega/2 - 6\delta/\eps$. This in turn implies that $\{\td(s, v)\}_j \ge \omega/2 - 6\delta/\eps - \omega/3 \ge 0$ and thus $\td'(s, v) < + \infty$. In particular, we can write
\begin{align}
\label{eq:almost_done}
    \td_i(s, v) - \td_j(s, v)
    &= \left( \lfloor \td(s, v) \rfloor_i + \hd_i(\sigma_s, v) - \omega \right)
    - \left( \lfloor \td(s, v) \rfloor_j + \hd_j(\sigma_s, v) - \omega \right)\\
    &= -\omega/3 + \hd_i(\sigma_s, v) - \hd_j(\sigma_s, v)
\end{align}
The value $\hd_i(\sigma_s, v)$ is bounded as 
\begin{align*}
    \hd_i(\sigma_s, v) 
    &\ge \ell_{H_i}(P)
    = \left(1 + \frac{\eps}{10^3}\right) \{\td(s, u')\}_i + \omega + \alpha(1+\eps)\ell_{G}(P')
\end{align*}
Similarly, we bound the value $\hd_j(\sigma_s, v)$ by
\begin{align*}
    \hd_j(\sigma_s, v) 
    &\le\left(1 + \frac{\eps}{10^4}\right)  \ell_{H_j}(P)
    = \left(1 + \frac{\eps}{10^4}\right) \left(
        \left(1 + \frac{\eps}{10^3}\right) \{\td(s, u')\}_j + \omega + \alpha(1+\eps)\ell_{G}(P')
    \right)
\end{align*}
Note that we have  $\{\td(s, u')\}_i \ge \omega/2 - 6\delta/\eps$ by \Cref{cl:contained_apsp}, hence $\{\td(s, u')\}_j \ge 0$ and, in particular, $\{\td(s, u)\}_j = \{\td(s, u)\}_i - \omega/3$. This allows us to simplify \Cref{eq:almost_done} as
\begin{align}
     \td_i(s, v) - \td_j(s, v)
    &= - \omega/3 +  \left(1 + \frac{\eps}{10^3}\right) \omega/3 - \frac{\eps}{10^4} \left(
        \left(1 + \frac{\eps}{10^3}\right) \{\td(s, u')\}_j + \omega + \alpha(1+\eps)\ell_{G}(P')
    \right)
\end{align}
We now only need to bound the right hand side error term as
\begin{align*}
        \left(1 + \frac{\eps}{10^3}\right) \{\td(s, u')\}_j + \omega + \alpha(1+\eps)\ell_{G}(P')
    &\le (1+ \eps/10^3) \omega + \omega + \alpha(1+\eps)\cdot 2\delta/(\alpha\eps)
    \le 3\omega
\end{align*}
and we conclude that $     \td_i(s, v) - \td_j(s, v) > 0$, a contradiction with our choice of $i$. 

\end{proof}

\section{Applications in Different Computational Models}
\label{sec:applications_computational_model}

In this section, we formally prove our main results in the PRAM and CONGEST model. Some of the results are slightly more general compared to their counterparts stated in the introduction. 
The results more or less straightforwardly follow by combining the reductions given in this paper together with the state-of-the-art parallel and distributed algorithms for computing approximate distance estimates. 
There are however some low-level technicalities that have to be dealt with and which we now discuss in more detail.

\subsection{Directed Graphs}

\Cref{alg:folklore_directed_smoothing} reduces the computation of exact distances in a graph with polynomially bounded integer edge weights to computing smoothly $2$-approximate distance estimates in graphs with nonnegative real weights.
Also, \Cref{alg:smoothing} reduces the computation of smoothly $2$-approximate distance estimates in graphs with real edge weights in $[1,\poly(n)]$ to the computation of $(1 + O(1/\log n))$-approximate distance estimates.

We would now like to combine these two reductions. That is, answer each oracle call in \Cref{alg:folklore_directed_smoothing} by using \Cref{alg:smoothing}.
The issue that has to be dealt with is that \Cref{alg:folklore_directed_smoothing}  invokes the smoothly $2$-approximate distance oracle on graphs with polynomially bounded maximum edge length, however, the minimum edge weight can be much smaller than one. In fact, the aspect ratio, i.e. the largest edge weight divided by the smallest edge weight, might not be polynomially boundeded.
However, this is merely a technicality.
One can easily adapt \Cref{alg:folklore_directed_smoothing} to ensure that it only invokes the distance oracle on graphs with polynomially bounded aspect ratio by adding $1/poly(n)$ to each edge when invoking the distance oracle, which ensure that the aspect ratio is polynomially bounded. Then, by scaling the weights one can additionally ensure that it only invokes the distance oracle on graphs with real edge weights in $[1,poly(n)]$. By performing these minor modifications, we obtain a reduction from computing exact distances in graphs with nonnegative polynomially bounded integer weights to $O(\log n)$ computations of $(1 + O(1/\log n))$-approximate distance estimates in graphs with nonnegative real weights.

We would now like to compute the $(1 + O(1/\log n))$-approximate distance estimates with the state-of-the art parallel and distributed algorithms. However, they only work in graphs with nonnegative polynomially bounded integer edge lengths.
One can slightly modify \Cref{alg:folklore_directed_smoothing} to ensure that it only invokes the approximate distance oracle on such graphs. First, we can ensure that it only invokes the approximate distance estimates on graphs with polynomial aspect ration by again adding an additive $1/poly(n)$ to each edge. Then, by discretizing and scaling the edge weights, we can ensure that the approximate distance oracle is only invoked on graphs with nonnegative polynomially bounded integer edge weights.
By implementing these minor modifications, we therefore obtain a reduction from computing exact distances in directed graphs with nonnegative polynomially bounded integer edge weights to $O(\log n)$ computations of $(1 + O(1/\log n))$-approximate distance estimates on directed graphs with nonnegative polynomially bounded integer edge weights.
Using the approximate distance algorithm of Cao, Fineman and Russell \cite{cao2020paralleldirected} and observing that all of the reductions can be efficiently implemented in the PRAM model, we therefore obtain the following result in the PRAM model.

\directedpram*

By observing that the reductions can also efficiently implemented in the \congest model, and that the state-of-the-art \congest algorithm for computing approximate distances in directed graphs also works if the source node is a virtual node, we obtain the following result in the \congest model.

\directedcongest*

\subsection{Undirected Graphs}

We already have argued that we can modify \Cref{alg:smoothing} to compute smoothly approximate distance estimates in graphs with polynomially bounded integer edge weights such that is only relies on computing approximate distance estimates in graphs with positive polynomially bounded integer edge weights. We can additionally modify it in such a way that it still preserves tree-likeness. We can therefore answer each oracle call of \Cref{alg:smoothing} using \Cref{alg:tree_constructing}. Moreover, we can also slightly modify \Cref{alg:tree_constructing}, along the lines we modified the other reductions, such that it only relies on approximate distance estimates on graphs equipped with positive polynomially bounded integer edge weights. By doing these modifications, we obtain a reduction for computing strongly $(1+\eps)$-approximate distance estimates in graphs with nonnegative polynomially bounded integer edge weights to computing $(1+\eps)$-approximate distance estimates in graphs with nonnegative polynomially bounded integer weights.

Using the state-of-the-art deterministic approximate distance algorithm for undirected graphs and observing that all of the reductions can be efficiently implemented in the PRAM model, we therefore obtain the following result in the PRAM model.

\undirectedpram*

By observing that the reductions can also efficiently implemented in the \congest model, and that the state-of-the-art \congest algorithm for computing approximate distances in undirected graphs also works if the source node is a virtual node, we obtain the following result in the \congest model.

\undirectedcongest*

\end{document}